\documentclass[final,5p,times,twocolumn]{elsarticle}
\usepackage{amssymb}
\usepackage{lineno}
\usepackage{amsmath}
\usepackage{cases}
\usepackage{array}
\usepackage{multirow}
\usepackage{booktabs} 
\usepackage{diagbox} 
\usepackage{xcolor}
\usepackage{threeparttable}
\usepackage{amsthm}
\usepackage{algpseudocode}
\usepackage{algorithm}
\usepackage{makecell}
\usepackage[normalem]{ulem}
\newtheorem{theorem}{Theorem}
\newtheorem{lemma}{Lemma}
\journal{Journal of Systems Architecture}

\begin{document}

\begin{frontmatter}

\title{Efficient Adaptive Bandwidth Allocation for Deadline-Aware Online Admission Control in Time-Sensitive Networking}

\author{Sifan Yu}
\ead{yusifan@buaa.edu.cn}
\author{Feng He}
\ead{robinleo@buaa.edu.cn}
\author{Anlan Xie}
\ead{xieanlan2000@buaa.edu.cn}
\author{Luxi Zhao\corref{cor1}}
\ead{zhaoluxi@buaa.edu.cn}
\cortext[cor1]{Corresponding author}

\address{School of Electronic and Information Engineering, Beihang University, China}

\begin{abstract}
With the growing demand for dynamic real-time applications, online admission control for time-critical event-triggered (ET) traffic in Time-Sensitive Networking (TSN) has become a critical challenge. The main issue lies in dynamically allocating bandwidth with real-time guarantees in response to traffic changes while also meeting the requirements for rapid response, scalability, and high resource utilization in online scenarios. To address this challenge, we propose an online admission control method for ET traffic based on the TSN/ATS+CBS (asynchronous traffic shaper and credit-based shaper) architecture. This method provides a flexible framework for real-time guaranteed online admission control, supporting dynamic bandwidth allocation and reclamation at runtime without requiring global reconfiguration, thus improving scalability. Within this framework, we further integrate a novel strategy based on network calculus (NC) theory for efficient and high-utilization bandwidth reallocation. On the one hand, the strategy focuses on adaptively balancing residual bandwidth with deadline awareness to prevent bottleneck egress ports, thereby improving admission capacity. On the other hand, it employs a non-trivial analytical result to reduce the search space, accelerating the solving process. Experimental results from both large-scale synthetic and realistic test cases show that, compared to the state-of-the-art, our method achieves an average 56\% increase in admitted flows and an average 92\% reduction in admission time. Additionally, it postpones the occurrence of bottleneck egress ports and the first rejection of admission requests, thereby enhancing adaptability.
\end{abstract}




\end{frontmatter}


\section{Introduction}
Industry 4.0 has accelerated the development of time-sensitive applications \cite{satka2023}, including industrial automation, smart grids, and automotive systems. These applications impose strict requirements for deterministic network transmissions, driving significant advancements in Time-Sensitive Networking (TSN). Moreover, with the increasing demand for dynamic applications like plug-and-play devices, there is an evolution from static to dynamic scenarios. In response to this trend, TSN integrates centralized network configuration (CNC) and centralized user configuration (CUC) through IEEE 802.1 Qcc \cite{Qcc2018} to enable dynamic configuration. However, despite IEEE 802.1 Qcc standardizing interfaces and protocols, achieving admission control for dynamic time-critical traffic remains an area with numerous open questions. The core challenge lies in implementing configurations that ensure real-time performance for dynamically changing traffic while also meeting the demands of online scenarios, including rapid response, strong scalability, and high resource utilization.

For time-triggered (TT) traffic implemented by scheduling mechanisms such as time-aware shaper (TAS) \cite{TAS2016} and cyclic queuing forwarding (CQF) \cite{CQF2015}, configuration involves the design of scheduling tables, which guarantee real-time requirements by incorporating deadline constraints \cite{feng2022a}. In order to reduce the computational complexity of static configuration methods, various online admission control methods have been developed for dynamic scenarios \cite{nayak2018, alnajim2019, huang2021a, feng2022, gartner2023, huang2023, cao2023}, enabling the flexible addition and removal of time-critical TT traffic. Nevertheless, for event-triggered (ET) traffic scheduled by mechanisms such as strict priority (SP) \cite{ieee8021q}, credit-based shaper (CBS) \cite{CBS2009}, and asynchronous traffic shaper (ATS) \cite{ATS2018}, configuration focuses on key parameters setting (e.g. bandwidth allocation), providing greater flexibility. However, ET traffic necessitates dedicated timing analysis \cite{deazua2014, ashjaei2017, zhao2021b} to guarantee real-time performance, ensuring that the worst-case delay does not exceed the end-to-end deadline. The study \cite{reusch2020} has shown that in traditional heuristic methods that rely on iterative optimization and timing analysis feedback validation, the timing analysis component consumes over 90\% of the configuration time, making them impractical for online scenarios. Although recent incremental performance analysis \cite{zhao2024a} for small configuration changes eliminates the need for global analysis at every configuration iteration and greatly reduces evaluation complexity, it still cannot avoid the need for feedback verification. When a massive number of flows change simultaneously, the complexity of incremental analysis may become nearly as high as that of global performance analysis, thereby diminishing its intended benefits. Zhao et al. \cite{zhao2024} have proposed a prior performance analysis method to reduce configuration time by avoiding iterative optimization and feedback loops. However, it requires considering the entire flow set and global egress ports, leading to high reconfiguration overhead and limited scalability in dynamic traffic scenarios. Guck et al. \cite{guck2016} and Maile et al. \cite{maile2022} have provided admission control methods for networks using SP and CBS, respectively. They ensure real-time requirements by setting delay budgets on queues and implementing delay-constrained least-cost (DCLC) routing across the network at runtime. However, this online DCLC routing increases admission time, and the fixed, unchangeable delay budgets established to handle burst cascades reduce the adaptability to dynamic flows, thereby impacting admission capacity. To address these issues, we propose a rapid, scalable, and high-utilization online admission control method for time-critical ET traffic based on the TSN/ATS+CBS architecture. Our main contributions are summarized as follows:
\begin{itemize}
	\item We propose a flexible online admission control framework with real-time guarantees. It supports dynamic resource allocation and reclamation to meet the real-time requirements of dynamic traffic, confining reconfiguration to small-scale adjustments and eliminating the need for global modifications, thus improving scalability.
		
	\item Meanwhile, within this framework, we propose an efficient and high-utilization resource adjustment strategy based on network calculus (NC) theory. It utilizes the concept of bandwidth balancing to avoid bottleneck links and applies a non-trivial analytical result during the solving process to reduce the search space, thereby enabling rapid bandwidth reallocation with high admission capacity.
	
	\item We evaluate our method against the state-of-the-art using both large-scale synthetic and realistic test cases. The results demonstrate that our method increases the number of admitted flows by an average of 56\% and reduces admission time by an average of 92\%. Additionally, it postpones the occurrence of bottleneck egress ports and the first rejection of admission requests, thus enhancing adaptability.
\end{itemize}

The rest of this paper is organized as follows. Section \ref{Background} describes the system model. Section \ref{Problem} outlines the problem statement and presents the overall solution. Section \ref{FlowAdd} introduces the flexible online admission control framework, and Section \ref{Strategy} proposes the deadline-adaptive local deadline adjustment strategy. Experimental evaluations are conducted in Section \ref{Experiment}. Finally, Section \ref{Conclusion} concludes the paper.

\section{System Model}
\label{Background}
\subsection{TSN Network and Flow Model}
In a TSN network, the network graph $\mathcal{G} = (\mathcal{N}, \mathcal{L})$ consists of a set of nodes $\mathcal{N}$ (including end-systems (ESs) and switches (SWs)) and a set of physical links $\mathcal{L}$. We define $(u,v) \in \mathcal{L}$ as the physical link from node $u$ to node $v$, also representing the corresponding egress port, with a transmission rate of $C$. The network supports both time-critical event-triggered (ET) traffic (also known as audio-video bridging (AVB) traffic) and non-time-critical best-effort (BE) traffic. An AVB flow $f \in \mathcal{F}$ is defined by the tuple $<s_{f},d_{f},l_{f},p_{f},D^{\rm E2E}_{f}>$, where $s_{f}$ is the source ES, $d_{f}$ is the destination ES, $l_{f}$ is the frame size, $p_{f}$ is the frame interval, and $D_{f}^{\rm E2E}$ is the end-to-end deadline. The maximum frame size for BE traffic is denoted as $l_{\rm BE}^{\max}$. A summary of the notations is provided in Table \ref{notation}.
\begin{figure}[!t]
	\centerline{\includegraphics[width=\columnwidth]{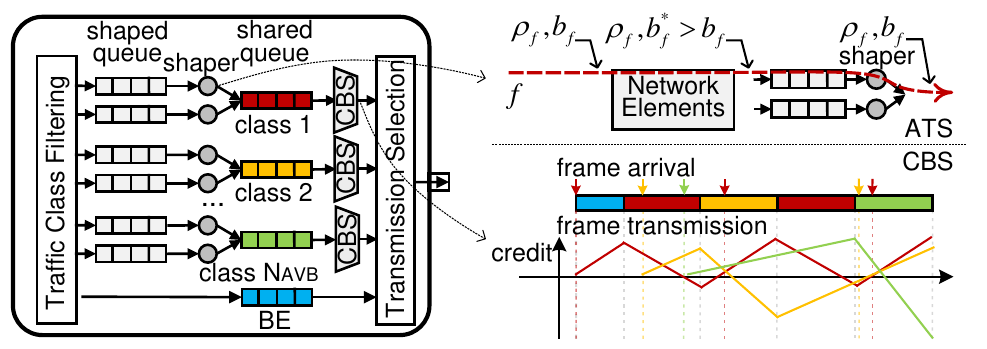}}
	\caption{TSN/ATS+CBS egress port architecture}
	\label{Model}
\end{figure}

\subsection{TSN/ATS+CBS Egress Port Architecture}
For each egress port $(u,v)\in \mathcal{L}$, this paper adopts a hybrid architecture that combines ATS \cite{ATS2018,specht2016} and CBS\cite{CBS2009}, referred to as TSN/ATS+CBS \cite{mohammadpour2018}, as depicted in Fig. \ref{Model}. This architecture performs per-class scheduling and provides eight shared queues. Among these, $N_{\rm AVB} \ (1 \leq N_{\rm AVB} \leq 8)$ queues are assigned to the time-critical AVB traffic classes $i \ (i \in [1,N_{\rm AVB}])$, while the remaining queues are available for non-time-critical BE traffic. Moreover, an AVB class $i$ is defined to have a higher priority than class $i+1$.

For AVB traffic, this hybrid architecture employs ATS to prevent burst cascades and CBS to provide flexible bandwidth reservation. Incoming AVB flows are initially directed to shaped queues according to the queuing schemes (QAR1, QAR2, QAR3) described in \cite{specht2017}. The ATS shaper, positioned after the shaped queues, enforces the committed burst size $b_{f} = l_{f}$ and the committed information rate $\rho_{f} = l_{f}/p_{f}$ for each flow $f$ by computing the eligible transmission time. The flows then proceed to the shared queues, where the CBS shaper manages traffic based on the reserved bandwidth $idSl^{(u,v)}_{i}$ (also known as the idle slope) for class $i$ at the egress port $(u,v)$. It maintains a credit value, initialized to zero, which decreases at a rate of $sdSl^{(u,v)}_{i} = idSl^{(u,v)}_{i} - C$ during frame transmission and increases at a rate of $idSl^{(u,v)}_{i}$ when frames are queued due to high-priority transmissions or negative credit. Additionally, the upper limit of bandwidth allocated to all AVB classes at the egress port, denoted by $idSl^{\max}$, is set by the designer but cannot exceed $C$ (e.g., $idSl^{\max} = 0.75C$). 

\subsection{Network Calculus-based Worst-Case Delay Analysis}
The network calculus (NC)-based worst-case delay analysis for TSN/ATS+CBS has been provided in \cite{mohammadpour2018}. In network calculus theory \cite{leboudec2001}, the arrival curve $\alpha(t)$ constrains the data bits accumulated from incoming flows over any time interval, while the service curve $\beta(t)$ represents the minimum processing capacity of the service element over any time interval. The maximum horizontal distance between the two curves is the worst-case delay bound at the egress port, as depicted in Fig. \ref{NCConcept} (a). In TSN/ATS+CBS, with the assistance of ATS reshaping, each AVB flow $f$ conforms to a committed burst size $b_{f} = l_{f}$ and a committed information rate $\rho_{f} = l_{f}/p_{f}$ at each egress port along its route. Then, the arrival curve of the aggregate flows of AVB class $i$ before the egress port $(u,v)$ is
\begin{equation}
\label{ArriveCurve}
\alpha_{i}^{(u,v)}(t) = \sum_{f\in \mathcal{F}_{i}^{(u,v)}}\left(\rho_{f}\cdot t + b_{f}\right),
\end{equation} 
where $\mathcal{F}_{i}^{(u,v)}$ is the subset of AVB flows $\mathcal{F}$ with class $i$ queuing at egress port $(u,v)$. The CBS service curve \cite{deazua2014} for AVB class $i$ at egress port $(u,v)$ is 
\begin{equation}
\beta_{i}^{(u,v)}(t) =
idSl_{i}^{(u,v)}\left[t-\frac{\sum_{j=1}^{i-1}\left(sdSl_{j}^{(u,v)}\cdot\frac{ l_{j}^{\max}}{C}\right)-l_{>i}^{\max}}{\sum_{j=1}^{i-1}idSl_{j}^{(u,v)}-C}\right]^{+},
\end{equation}
where $l_{i}^{\max}$ is the maximum frame size of AVB class $i$, and $l_{>i}^{\max} = \max_{j>i}\{l_{j}^{\max},l_{\rm BE}^{\max}\}$ is the maximum frame size of traffic with lower priorities than $i$. The upper bound of the worst-case delay for any flow $f \in \mathcal{F}_{i}^{(u,v)}$ queuing at egress port $(u,v)$ is the maximum horizontal distance between the arrival curve $\alpha_{i}^{(u,v)}(t)$ and the service curve $\beta_{i}^{(u,v)}(t)$. Furthermore, it is proven in \cite{leboudec2018} that ATS shaping after a first-in, first-out (FIFO) system does not add extra delays to the worst-case delay of the combination. Therefore, AVB flows do not experience extra worst-case delays in the reshaping queues before ATS. The worst-case delay bound $H_{i}^{(u,v)}$ for AVB class $i$ queuing at egress port $(u,v)$, as shown in Fig. \ref{NCConcept} (b), can ultimately be expressed as
\begin{equation}
\label{WorstDelayFormula}
H_{i}^{(u,v)} = \frac{\sum_{f \in \mathcal{F}_{i}^{(u,v)}} b_{f}}{idSl_{i}^{(u,v)}} +
\frac{l^{\max}}{C} + \frac{(i-1)l^{\max}}{C-\sum_{j=1}^{i-1}idSl_{j}^{(u,v)}},
\end{equation}
where $l_{i}^{\max}$ is simplified to $l^{\max}$, representing the maximum frame size in the network. In this way, Eq. (\ref{WorstDelayFormula}) remains valid even if the maximum frame sizes of other classes vary due to dynamic traffic changes.
\begin{figure}[!t]
	\centerline{\includegraphics[width=\columnwidth]{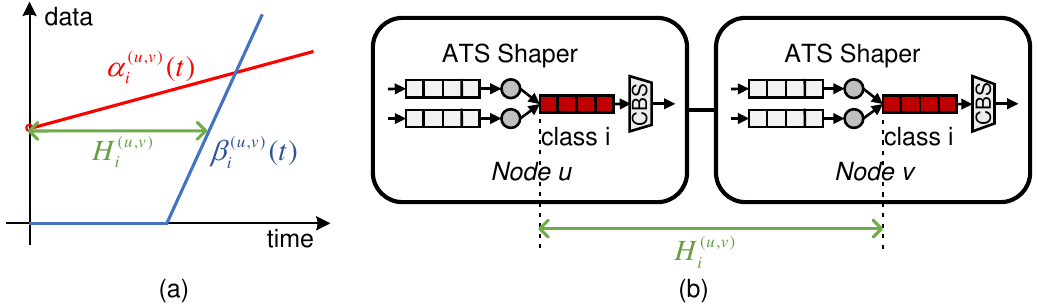}}
	\caption{(a) Concepts in network calculus theory, (b) worst-case delay bound in TSN/ATS+CBS.}
	\label{NCConcept}
\end{figure}

\begin{table}[!t]
	\caption{Summary of notations}
	\label{notation}
	\scriptsize
	\setlength{\tabcolsep}{3pt}
	\begin{footnotesize}
		\begin{tabular}{|p{35pt}|p{195pt}|}
			\hline
			Symbol&Meaning\\
			\hline
			$\mathcal{G}$&Network graph\\
			$\mathcal{N}$&Set of nodes\\
			$\mathcal{L}$&Set of physical links\\
			$(u,v)$&Physical link from node $u$ to node $v$, also representing the corresponding egress port\\
			$C$&Transmission rate of physical link\\
			$f$&An AVB flow\\
			$s_{f}$&Source ES of flow $f$\\
			$d_{f}$&Destination ES of flow $f$\\
			$l_{f}$&Frame size of flow $f$\\
			$p_{f}$&Frame interval of flow $f$\\
			$D^{\rm E2E}_{f}$&End-to-end deadline of flow $f$\\
			$b_{f}$&Committed burst size of flow $f$\\
			$\rho_{f}$&Committed information rate of flow $f$\\
			$r_{f}$&Route of flow $f$\\
			$N_{\rm AVB}$&Number of AVB classes supported by each egress port\\
			$\mathcal{F}$&Set of AVB flows\\
			$\mathcal{F}_{i}^{(u,v)}$&Set of AVB flows of class $i$ at egress port $(u,v)$\\
			$idSl^{\max}$&Upper limit of bandwidth allocated to all AVB classes at the egress port\\
			$\alpha_{i}^{(u,v)}$&Arrival curve of the aggregate flows of AVB class $i$ before the egress port $(u,v)$\\
			$\beta_{i}^{(u,v)}$&Service curve for AVB class $i$ at the egress port $(u,v)$\\
			$l_{i}^{\max}$&Maximum frame size of AVB class $i$\\
			$l_{>i}^{\max}$&Maximum frame size of traffic with lower priorities than $i$\\
			$l^{\max}$&Maximum frame size in the network\\
			$l_{\rm BE}^{\max}$&Maximum frame size of BE traffic\\
			$D_{i}^{(u,v)}$&Existing local deadline for AVB class $i$ at egress port $(u,v)$ in the existing configuration\\
			$\hat{D}_{i}^{(u,v)}$&Adjusted local deadline for AVB class $i$ at egress port $(u,v)$ in the flow addition process\\
			$\check{D}_{i}^{(u,v)}$&Adjusted local deadline for AVB class $i$ at egress port $(u,v)$ in the flow removal process\\
			$idSl_{i}^{(u,v)}$&Existing idle slope, i.e. bandwidth for AVB class $i$ at egress port $(u,v)$ in the existing configuration\\
			$\hat{idSl}_{i}^{(u,v)}$&Adjusted idle slope, i.e. bandwidth for AVB class $i$ at egress port $(u,v)$ in the flow addition process\\
			$\check{idSl}_{i}^{(u,v)}$&Adjusted idle slope, i.e. bandwidth for AVB class $i$ at egress port $(u,v)$ in the flow removal process\\
			$D_{f}^{(u,v)}$&Local deadline for flow $f$ at egress port $(u,v) \in r_{f}$\\
			$\mathcal{R}(s,d)$&Candidate route set for source-destination pair, where $s$ and $d$ are end-systems\\
			$R^{(u,v)}$&Available residual bandwidth for AVB traffic at the egress port $(u,v)$\\
			$\Phi_{i}$&Extra bandwidth allocated to AVB class $i$\\
			\hline
		\end{tabular}
	\end{footnotesize}
\end{table}

\begin{figure*}[!t]
	\centerline{\includegraphics[width=\textwidth]{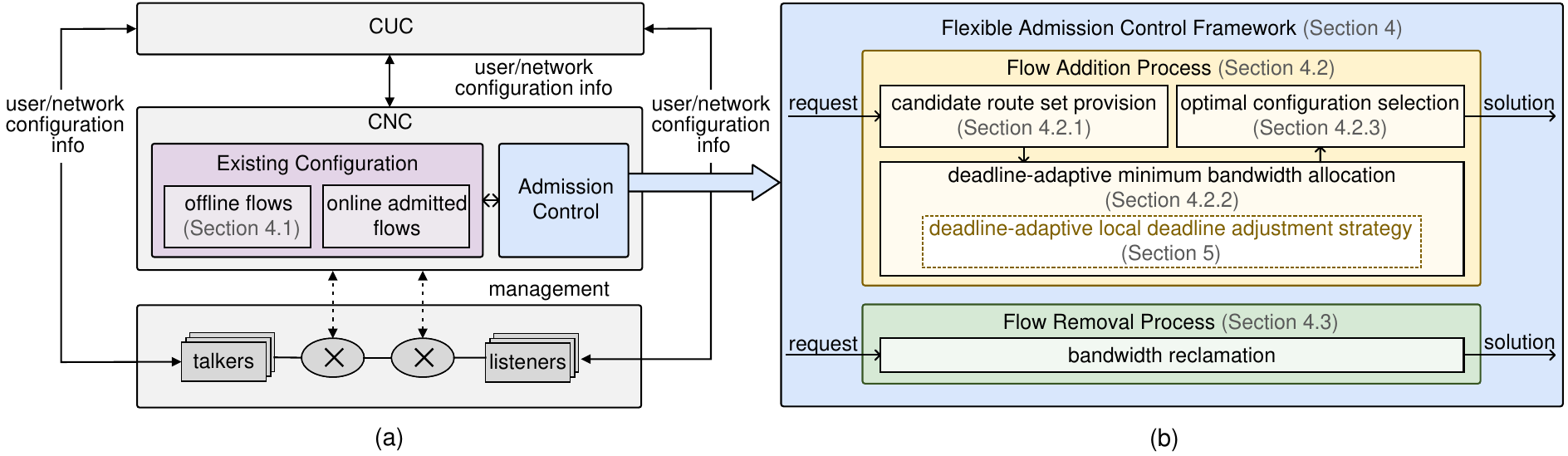}}
	\caption{(a) A fully centralized model, (b) flexible online admission control framework.}
	\label{FrameworkCNC}
\end{figure*}
\section{Problem Statement and Overall Solution}
\label{Problem}
\subsection{Problem Statement}
A fully centralized model that supports dynamic traffic changes comprises two main components: the centralized user configuration (CUC) and the centralized network configuration (CNC), as illustrated in Fig. \ref{FrameworkCNC}(a). The CUC collects flow requests from end systems and forwards them to the CNC, which is responsible for managing existing configurations and executing admission control. When a new flow is requested, the CNC evaluates the available resources to determine if the flow can be admitted and, if so, performs the necessary resource allocation. When an existing flow is removed, the CNC reclaims the resources. This resource allocation and reclamation should not only satisfy the end-to-end deadline requirements for all flows but also address the needs of online admission control scenarios, including rapid response, scalability, and high resource utilization.

Current resource allocation methods for ET traffic in TSN still have limitations in online admission control scenarios. Conventional heuristic methods that use iterative optimization and timing analysis feedback validation can satisfy end-to-end deadline requirements. However, their high computational complexity hinders rapid response in online admission control scenarios. \cite{zhao2024} decomposes global end-to-end deadlines into local deadlines at individual ports and allocates only the minimum bandwidth required to satisfy these local deadlines. This approach avoids repeated global timing analysis, thereby significantly reducing computational complexity. However, since it must account for complete flows and global egress ports, frequent flow changes incur high reconfiguration overhead, thereby limiting scalability in dynamic large-scale networks. \cite{maile2022} pre-configures fixed delay budgets for egress ports to confine reconfiguration to only the specific routes of changing flows. However, these fixed delay budgets can lead to imbalanced resource allocation, thereby reducing overall resource utilization. Therefore, there is an urgent need for a dynamic resource allocation approach for admission control of time-critical ET traffic that not only meets end-to-end requirements but also provides rapid responsiveness, strong scalability, and high bandwidth utilization.

\subsection{Overall Admission Control Solution}
To address these issues, we propose an online admission control framework for time-critical ET traffic based on the TSN/ATS+CBS architecture, as shown in Fig. \ref{FrameworkCNC}(b). It establishes flexible local deadlines to enforce port-level real-time constraints on traffic transmission, supporting low-overhead runtime adjustments to both local deadlines and bandwidth allocation to provide real-time guarantees for dynamically changing traffic.

Before the system operates, the initial flows are configured in an offline scenario (Section \ref{InitialConf}). When a new flow is requested to be added (Section \ref{FlowAddition}), the method evaluates every route in the candidate set. For each candidate route, it employs an efficient, high-utilization strategy to calculate the local deadline adjustments required to accommodate the new flow (Section \ref{Strategy}), followed by the computation of the minimum bandwidth necessary to meet these adjusted local deadlines. Finally, the candidate route and its associated bandwidth allocation, yielding the lowest overall bandwidth consumption, are selected as the admission scheme for the new flow. When an existing flow is removed (Section \ref{FlowRemoval}), the method adjusts the local deadlines and bandwidth allocations along its route to reclaim resources, thus conserving bandwidth for future flow admissions and non-real-time traffic transmission.

Compared to current approaches, our method incorporates several unique design features that make it well-suited for online admission control of time-critical ET traffic in dynamic scenarios, as detailed below:

(1) \textit{Flexible Adjustable Local Deadline:}
Local deadlines serve as critical constraints on port-level bandwidth allocation. Hence, their setting has a significant impact on overall bandwidth utilization. In contrast to methods that rely on fixed delay budgets, such as \cite{maile2022}, our method supports flexible adjustable local deadlines at each egress port. This design enhances allocation flexibility and promotes improved bandwidth utilization.

\begin{figure}[!t]
\centerline{\includegraphics[width=\columnwidth]{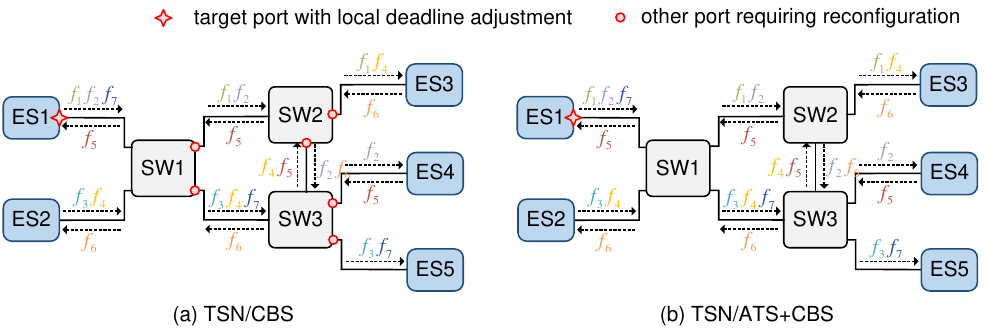}}
\caption{Comparison of reconfiguration overhead caused by local deadline adjustment in TSN/CBS and TSN/ATS+CBS.}
\label{Adjustment}
\end{figure}
(2) \textit{Reducing Reconfiguration Overhead:}
Local deadline adjustments in typical TSN/CBS architecture can cause high reconfiguration overhead, as illustrated in Fig. \ref{Adjustment}. When the local deadline at an egress port increases, the jitter of flows at that port also grows, leading to more pronounced burstiness at subsequent ports. Thus, all subsequent ports must be reconfigured to ensure that their local deadlines are not violated, thus causing high reconfiguration overhead. To address this issue, we integrated ATS into TSN/CBS to eliminate burst cascades, thereby preventing the effects of local deadline adjustments from propagating to other ports. This design decouples the configuration of egress ports, effectively confining the reconfiguration scope and reducing overhead.

(3) \textit{Efficient and High-uitilization Resource Reallocation:}
In response to traffic changes, we aim to rapidly adjust bandwidth allocation to meet the end-to-end deadlines of ET traffic while maximizing bandwidth utilization. Although \cite{zhao2024} provides an approach for deriving the minimum required bandwidth at the port level from given local deadlines, determining appropriate local deadlines for each port remains an open problem. Therefore, we propose an efficient and high-utilization local deadline adjustment strategy. It aims to balance residual bandwidth to avoid bottlenecks, thereby improving admission capacity. Furthermore, it formally derives a non-trivial coupling relationship among local deadlines under balanced residual bandwidth conditions, reducing the solution space and enabling rapid solving.

\section{Flexible Online Admission Control Framework}
\label{FlowAdd}
\subsection{Configuration Before Admission Control}
\label{InitialConf}
Before the online admission control at runtime, the existing configuration should provide real-time guarantees for offline flows. This configuration is based on a QoS-based minimum bandwidth allocation analytical method \cite{zhao2024} and is carried out in two stages:

The first stage is to determine the local deadline for each class at every egress port, ensuring that the end-to-end deadline requirements of all offline flows are met. For any egress port $(u,v)$ and AVB class $i$, the local deadline $D^{(u ,v)}_{i}$ can be obtained by
\begin{equation}
\label{S1}
\begin{aligned}
&Stage \ 1:\forall (u,v) \in \mathcal{L}, i \in [1,N_{\rm AVB}]:  \quad\quad\quad\quad\quad\quad\quad\quad\quad\quad\quad\quad \\
&\quad\quad\quad\quad\quad\quad\qquad D_{i}^{(u,v)} =\min_{f\in \mathcal{F}_{i}^{(u,v)}}\Big\{D_{f}^{(u,v)}\Big\},\\
\end{aligned}
\end{equation}
where
\begin{equation}
\label{flowdeadline}
\begin{aligned}
\sum_{(u,v) \in r_{f}} D_{f}^{(u,v)} \leq D_{f}^{\rm E2E},\\
\end{aligned}
\end{equation}
where $\mathcal{F}^{(u,v)}_{i}$ denotes the set of existing flows (currently only offline flows) with class $i$ queuing at egress port $(u,v)$, $r_{f}$ denotes the route of flow $f$. $D_{f}^{(u,v)}$ denotes the local deadline of flow $f$ at egress port $(u,v) \in r_{f}$, with the constraint that the sum of local deadlines along its path does not exceed its end-to-end deadline. $D_{f}^{(u,v)}$ for offline flows can be obtained by \cite{zhao2024}. In order to meet the strictest deadlines, the local deadline for each class at every egress port in Eq. (\ref{S1}) is set to the minimum of the local deadlines for all flows within that class at that port. 

The second stage is to allocate the minimum bandwidth for each class at every egress port to satisfy the local deadlines. This stage is defined as
\begin{equation}
\begin{aligned}
\label{MiniBand}
&Stage \ 2:\forall (u,v) \in \mathcal{L}, i \in [1,N_{\rm AVB}]:\\
&idSl_{i}^{(u,v)} = 
\max\Bigg\{\underbrace{\frac{\sum_{f\in \mathcal{F}_{i}^{(u,v)} }b_{f} }{D_{i}^{(u,v)}-\frac{l^{\max}}{C} - \frac{(i-1)l^{\max}}{C-\sum_{j=1}^{i-1}idSl_{j}^{(u,v)}}}}_{first},
\underbrace{\sum_{f \in \mathcal{F}_{i}^{(u,v)}}\rho_{f} }_{second}\Bigg\},\\
\end{aligned}
\end{equation}
where the first term determines the minimum bandwidth required to satisfy the local deadline $D_{i}^{(u,v)}$ for each class $i \in [1, N_{\rm AVB}]$ at every egress port $(u,v) \in \mathcal{L}$, as proven in \cite{zhao2024}, the second term maintains network stability. Note that since we leverage ATS in the admission control architecture to avoid burst cascades, the burst size $b_{f}$ and long-term rate $\rho_{f}$ of flow $f$ at the egress port $(u,v)$ in Eq. (\ref{MiniBand}) remain the same as they were when sent from the source ES. Additionally, Eq. (\ref{MiniBand}) reveals that the minimum bandwidth allocation to lower priority classes depends on the allocation results of higher priority classes. Therefore, the minimum bandwidth should be allocated from high priority to low priority.

Certain network configuration information needs to be stored in the CNC, which will serve as a reference for online admission control. For each flow $f \in \mathcal{F}$, its tuple, route $r_{f}$, and the local deadlines $D_{f}^{(u,v)}$ along $r_{f}$ should be stored. Additionally, for each egress port $(u,v) \in \mathcal{L}$ and each class $i \in [1, N_{\rm AVB}]$, the local deadline $D_{i}^{(u,v)}$ and the bandwidth allocation $idSl_{i}^{(u,v)}$ should also be stored. Subsequently, we will explore how to perform online admission control relying on the existing configuration. We aim to maximize the number of admitted flows with minimized bandwidth usage while ensuring that the deadline requirements of both new and existing flows are met, by quickly and dynamically allocating and reclaiming bandwidth resources. The updated configuration will replace the existing one and be stored in the CNC for the next admission control operation.

\subsection{Flow Addition Process}
\label{FlowAddition}
\begin{figure}[!t]
	\centerline{\includegraphics[width=\columnwidth]{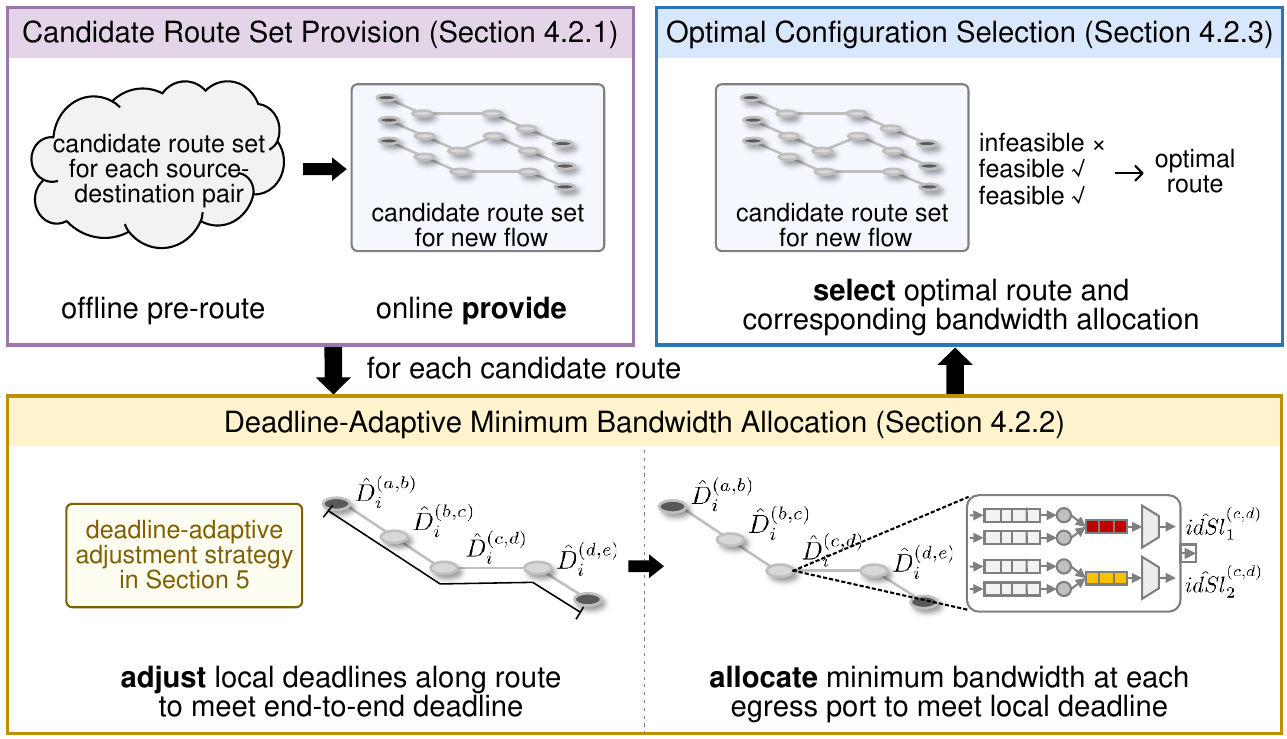}}
	\caption{Flow Addition Process.}
	\label{FlowAddionFig}
\end{figure}
When a new flow $f^{+}$ requests to be added, the flow addition process in our admission control framework evaluates, based on the available bandwidth, whether to admit the request and how to provide configuration if admitted. The configuration, including the route and allocated bandwidth, should satisfy the end-to-end deadline requirements of both the new flow $f^{+}$ and existing flows (including both offline flows and online admitted flows) while maximizing the utilization of remaining network resources. The specific procedure of the flow addition process is shown in Fig. \ref{FlowAddionFig}. Given that online routing is time-consuming, the candidate route set for each source-destination pair is generated offline through pre-routing. During the online phase, a new flow is assigned the corresponding candidate route set based on its source ES and destination ES (Section \ref{PreRouting}). Subsequently, the process provides a deadline-adaptive minimum bandwidth allocation scheme for each route within the candidate set of the new flow (Section \ref{PathLevelAjust}). This scheme aims to quickly meet end-to-end deadline requirements while adaptively balancing the residual bandwidth (further detailed in Section \ref{Strategy}). Finally, this process evaluates all the routes within the candidate set of the new flow along with their corresponding bandwidth allocation schemes and selects the optimal configuration (Section \ref{PathSelection}). Detailed discussions will be provided as follows.

\subsubsection{Candidate Route Set Provision}
\label{PreRouting}
When a new flow $f^{+}$ requests to be added, we first need to provide it with a route. To avoid the complexity of online routing, we pre-route a set of candidate routes for each source-destination pair, enabling efficient online retrieval. We use the k-shortest path algorithm \cite{yen1971} to generate the candidate route set ${\mathcal{R}(s,d)}$ for each source-destination pair, where $s$ and $d$ are end-systems. We choose the k-shortest path algorithm for two main reasons. On the one hand, we aim to select shorter routes to achieve larger local deadlines for each egress port along the route. This ensures that more residual bandwidth is reserved according to Eq. (\ref{MiniBand}), providing greater flexibility in bandwidth allocation while maintaining end-to-end deadline performance (Section \ref{PathLevelAjust}). On the other hand, providing k routes offers greater flexibility in route selection, thereby better balancing the remaining network bandwidth (Section \ref{PathSelection}). During the online phase, a new flow $f^{+}$ is assigned the corresponding candidate route set $\mathcal{R}(s_{f^{+}}, d_{f^{+}})$ based on its source ES $s_{f^{+}}$ and destination ES $d_{f^{+}}$.

\subsubsection{Deadline-Adaptive Minimum Bandwidth Allocation}
\label{PathLevelAjust}
When the new flow $f^{+}$ is assigned to a candidate route $r_{f^{+}} \in \mathcal{R}(s_{f^{+}}, d_{f^{+}})$, we need to find a bandwidth allocation scheme ($\hat{idSl}$) that meets the end-to-end deadline requirements. This procedure is generally similar to the stages used for configuring offline flows in Section \ref{InitialConf}, with the key difference being in Stage 1, where local deadlines are determined. Since adding the new flow may render existing local deadlines inadequate, we focus on adjusting them to meet the updated end‑to‑end requirements. In order to minimize changes to the existing configuration, we adjust the local deadlines only at the egress ports along the route $r_{f^{+}}$ of the new flow with the assistance of ATS rather than reallocating local deadlines for all egress ports in the network. Once the local deadlines are set, we can directly apply the method from Stage 2 to calculate the bandwidth allocation scheme for adding flow $f^{+}$ along the candidate route $r_{f^{+}}$.

To meet the end-to-end deadline $D_{f^{+}}^{\rm E2E}$ of the new flow $f^{+}$, we define the adjusted local deadline $\hat{D}_{i}^{(u,v)}$ for class $i$ (to which $f^{+}$ belongs) at each egress port $(u,v) \in r_{f^{+}}$ as
\begin{equation}
\label{adjustedD}
\begin{aligned}
\hat{D}_{i}^{(u,v)} = 
\begin{cases}
D_{i}^{(u,v)}, &\sum\limits_{(u,v) \in r_{f^{+}}}D_{i}^{(u,v)} \leq D_{f^{+}}^{\rm E2E}\\
\rm result \ \rm from \ Alg. 1, &\sum\limits_{(u,v) \in r_{f^{+}}}D_{i}^{(u,v)} > D_{f^{+}}^{\rm E2E}\\
\end{cases},
\end{aligned}
\end{equation}
where $D^{(u,v)}_{i}$ is the existing local deadline for class $i$ at egress port $(u,v)$, and $D^{\rm E2E}_{f^{+}}$ is the end-to-end deadline of the new flow $f^{+}$. There are two cases for adjusting the local deadlines:

i) \textit{No Adjustment Required:} When the sum of the existing local deadlines $D^{(u,v)}_{i}$ along the route $r_{f^{+}}$ meets the end-to-end deadline of the new flow $f^{+}$, i.e., $\sum_{(u,v) \in r_{f^{+}}}D_{i}^{(u,v)} \leq D_{f^{+}}^{\rm E2E}$, there is no need for further adjustments to the local deadlines to meet the end-to-end deadline requirements. As a result, the local deadline settings remain unchanged, i.e., $\hat{D}^{(u,v)}_{i} = D^{(u,v)}_{i}$ for $(u,v) \in r_{f^{+}}$. 

ii) \textit{Adjustment Required:} When the sum of the existing local deadlines $D^{(u,v)}_{i}$ along the route $r_{f^{+}}$ exceeds the end-to-end deadline of the new flow $f^{+}$, i.e., $\sum_{(u,v) \in r_{f^{+}}}D_{i}^{(u,v)} > D_{f^{+}}^{\rm E2E}$, $D^{(u,v)}_{i}$ at each egress port $(u,v) \in r_{f^{+}}$ needs to be reduced to meet the end-to-end deadline of new flow $f^{+}$. However, adjusting $D^{(u,v)}_{i}$ is not straightforward, as reducing these local deadlines will affect the distribution of the remaining bandwidth across different egress ports in the network, requiring careful balancing to optimize resource utilization. We determine $\hat{D}_{i}^{(u,v)}$ using Algorithm 1, specifically discussed in Section \ref{DelayAdjust}, by employing a deadline-adaptive local deadline adjustment strategy to balance the remaining bandwidth resources. 

After obtaining the adjusted local deadlines $\hat{D}^{(u,v)}_{i}$ for class $i$ at each egress port $(u,v) \in r_{f^{+}}$, the corresponding bandwidth allocation $\hat{idSl}_{j}^{(u,v)} \ ( j \in [1,N_{\rm AVB}], (u,v) \in \mathcal{L})$ for adding flow $f^{+}$ on candidate route $r_{f^{+}}$ is determined. This is done by substituting existing flow set $\mathcal{F}$ with $\mathcal{F} \cup \{f^{+}\}$ and $D_{i}^{(u,v)}$ with $\hat{D}_{i}^{(u,v)}$ and then calculating using Eq. (\ref{MiniBand}). According to Eq. (\ref{MiniBand}), the addition of flow $f^{+}$ only affects the bandwidth allocation for classes with priorities equal to or lower than that of flow $f^{+}$ on route $r_{f^{+}}$, without causing large-scale reconfiguration. Additionally, the local deadline $D_{f^{+}}^{(u,v)}$ for new flow $f^{+}$ along the route $r_{f^{+}}$ is assigned the value of $\hat{D}_{i}^{(u,v)}$, ensuring that $\sum_{(u,v) \in r_{f^{+}}} D_{f^{+}}^{(u,v)} \leq D_{f^{+}}^{\rm E2E}$ is satisfied in accordance with Eq. (\ref{flowdeadline}) from Stage 1.

\subsubsection{Optimal Configuration Selection}
\label{PathSelection}
Finally, we obtain the optimal configuration by selecting the optimal route $r_{f^{+}}^{\rm opt}$ from the $k$ candidate routes $r_{f^{+}} \in \mathcal{R}(s_{f^{+}}, d_{f^{+}})$ based on the corresponding bandwidth allocation ($\hat{idSl}$) obtained from Section \ref{PathLevelAjust}. Initially, we need to evaluate the feasibility of each candidate route, which is determined by two conditions. The first condition is that feasible local deadlines can be provided by Eq. (\ref{adjustedD}). The second condition is that the sum of bandwidth allocations at each egress port does not exceed the upper limit $idSl^{\max}$, i.e., $\sum_{i=1}^{N_{\rm AVB}} \hat{idSl}_{i}^{(u,v)} \leq idSl^{\max}$. If at least one feasible route exists, we will select the optimal one to admit the new flow $f^{+}$ using the following method. If no feasible route exists, it is impossible to find a route that meets the end-to-end deadline requirement of the new flow, and the request to add the new flow will be rejected.

The optimal route $r_{f^{+}}^{\rm opt}$ is defined as
\begin{equation}
\label{route}
r_{f^{+}}^{\rm opt} = \arg \min_{r_{f^{+}} \in \mathcal{R}(s_{f^{+}},d_{f^{+}}) \land \rm feasible} \{cost(r_{f^{+}})\},
\end{equation} 
and
\begin{equation}
\label{CostEq}
cost(r_{f^{+}}) = 
\sum_{(u,v)\in \mathcal{L}}\Bigg(\frac{1}{idSl^{\max} - \sum_{i=1}^{N_{\rm AVB}}\hat{idSl}_{i}^{(u,v)}}- \frac{1}{idSl^{\max}}\Bigg)^{2},\\
\end{equation}
where $cost(r_{f^{+}})$ represents the network cost associated with the bandwidth allocation scheme for adding flow $f^{+}$ along the candidate route $r_{f^{+}}$. Considering its residual bandwidth, the port cost of $(u,v)$ is defined as $1/(idSl^{\max} - \sum_{i=1}^{N_{\rm AVB}}\hat{idSl}_{i}^{(u,v)})$. When no bandwidth is utilized, this cost reaches its minimum value of $1/idSl^{\max}$. The network cost $cost(r_{f^{+}})$ in Eq. (\ref{CostEq}) is determined by calculating the difference between the actual port cost and the minimum port cost, squaring this difference, and summing the results of all egress ports. When the residual bandwidth of an egress port is low, the network cost $cost(r_{f^{+}})$ rises sharply, indicating bottleneck egress ports nearing resource exhaustion, potentially hindering future flow admissions. Consequently, we can select the optimal route $r_{f^{+}}^{\rm opt}$ by minimizing the network cost as defined in Eq. (\ref{CostEq}), which balances the residual bandwidth. This helps avoid bottleneck egress ports caused by resource exhaustion, thereby enhancing the capacity of the network to accommodate future traffic.

\subsection{Flow Removal Process}
\label{FlowRemoval}
After removing an existing flow $f^{-}$, the flow removal process needs to provide a bandwidth allocation scheme ($\check{idSl}$) to reclaim bandwidth resources for future traffic access. First, we remove the local deadlines related to flow $f^{-}$. The adjusted local deadline $\check{D}_{i}^{(u,v)}$ for class $i$ (to which $f^{-}$ belongs) at egress port $(u,v) \in r_{f^{-}}$ is
\begin{equation}
\check{D}_{i}^{(u,v)} = \min_{f\in (\mathcal{F} \setminus \{f^{-}  \})_{i}^{(u,v)}} \Big\{D_{f}^{(u,v)}\Big\}.
\end{equation}
Subsequently, the adjusted bandwidth allocation $\check{idSl}_{j}^{(u,v)} \ ( j \in [1,N_{\rm AVB}], (u,v) \in \mathcal{L})$ for removing flow $f^{-}$ is determined. This is done by substituting existing flow set $\mathcal{F}$ with $\mathcal{F} \setminus \{f^{-}\}$ and $D_{i}^{(u,v)}$ with $\check{D}_{i}^{(u,v)}$ and then calculating in Eq. (\ref{MiniBand}). Similar to flow addition, the bandwidth reclamation after removing flow $f^{-}$ affects only the bandwidth allocation for classes with priorities equal to or lower than that of flow $f^{-}$ on route $r_{f^{-}}$, without the need for large-scale reconfiguration.

\section{Deadline-Adaptive Local Deadline Adjustment Strategy}
\label{Strategy}
\subsection{Challenges in Local Deadline Adjustment}
As mentioned in Section \ref{PathLevelAjust}, if the sum of the existing local deadlines $D_{i}^{(u,v)}$ for class $i$ (to which $f^{+}$ belongs) along route $r_{f^{+}}$ exceeds the end-to-end deadline $D_{f^{+}}^{\rm E2E}$, $D_{i}^{(u,v)}$ for each $(u,v) \in r_{f^{+}}$ should be adjusted to meet the end-to-end deadline. Since local deadlines serve as constraints on port-level bandwidth allocation, their adjustment plays a crucial role in overall network bandwidth utilization. However, developing an efficient and high-utilization local deadline adjustment strategy is quite challenging. If the local deadlines along the route are treated as independent decision variables, the problem becomes a multivariable combinatorial optimization task, which generally entails high computational costs. Some studies \cite{zotero-5194} have directly modified the local deadlines $D_{i}^{(u,v)}$, adjusting them either equally or proportionally. Although these approaches are simple and computationally efficient, they do not balance residual bandwidth, which can lead to early resource exhaustion at bottleneck egress ports and thus reduce admission capacity.

\begin{figure}[!t]
	\label{DelayAdjust}
	\centerline{\includegraphics[width=\columnwidth]{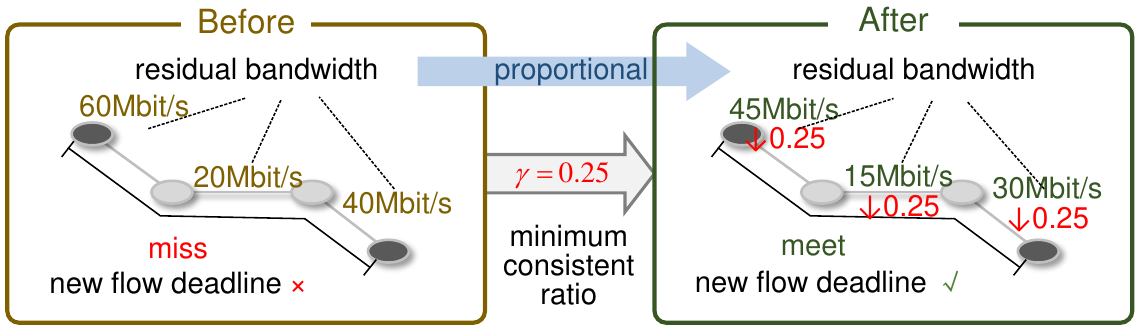}}
	\caption{An example of local deadlines adjusted by minimum consistent ratio $\gamma = 0.25$ to meet the new flow deadline.}
	\label{AdjustFigure}
\end{figure}

\subsection{Local Deadline Adjustment Strategy}
In this section, we propose a local deadline adjustment strategy, as illustrated in Fig. \ref{AdjustFigure}. This strategy aims to adjust local deadlines to meet end-to-end deadlines while maximizing resource utilization by reducing bandwidth consumption and balancing residual bandwidth. Specifically, we define $\gamma$ as the ratio of extra bandwidth an egress port receives in the adjustment to its original residual bandwidth, and require that every port along $r_{f^+}$ adopt the same $\gamma$. This can synchronize resource usage across egress ports and helps prevent the premature emergence of bottleneck links. Building on this, we further derive an analytical expression that maps $\gamma$ to the corresponding adjusted local deadline $\hat{D}_{i}^{(u,v)}$ for each egress port. This mapping effectively tackles the complexity introduced by high‑priority bandwidth adjustments on lower‑priority traffic and uniquely determines $\hat{D}_{i}^{(u,v)}$ for each $(u,v) \in r_{f^{+}}$ through analytical computation based on $\gamma$. As a result, all local deadlines are now coupled through $\gamma$, reducing the original multivariable combinatorial optimization problem to a single‑variable problem of determining $\gamma$. Finally, recognizing the nonlinearity of the mapping, we employ a binary search to find the minimum $\gamma$ that satisfies the end‑to‑end deadline requirements, thereby minimizing bandwidth usage. The process for obtaining this minimum $\gamma$ and the corresponding updated local deadlines $\hat{D}_{i}^{(u,v)}$ for each $(u,v) \in r_{f^{+}}$ is detailed in the following algorithm.

\textbf{Algorithm:} 
We propose a local deadline adjustment algorithm, with its pseudo-code presented in Algorithm 1. The algorithm first calculates the available residual bandwidth $R^{(u,v)}$ for each $(u,v) \in r_{f^{+}}$ using the \texttt{CalResBand} function described in Section \ref{CalculateRes} (lines 1-3). Then, the algorithm calculates the extra bandwidth $\gamma \cdot R^{(u,v)}$ needed to meet the end-to-end deadline of the new flow, thereby reducing the local deadlines $D_{i}^{(u,v)}$ along the route $r_{f^{+}}$ by allocating this extra bandwidth. As indicated in Eq. (\ref{MiniBand}), achieving a smaller adjusted local deadline requires a larger bandwidth allocation. Initially, assuming all available residual bandwidth $R^{(u,v)}$ is used as extra bandwidth, the algorithm calculates the minimum possible adjusted local deadlines $\hat{D}_{i}^{(u,v),\rm min}$ for each $(u,v) \in r_{f^{+}}$ using the \texttt{MapBand2D} function described in Section \ref{SectionMap} (lines 4-6). If the sum of the minimum possible adjusted local deadlines $\hat{D}_{i}^{(u,v),\rm min}$ along the route $r_{f^{+}}$ cannot meet the end-to-end deadline $D_{f^{+}}^{\rm E2E}$, it indicates that no feasible local deadline adjustment exist (lines 18). Conversely, if the sum is smaller than the end-to-end deadline $D_{f^{+}}^{\rm E2E}$, it means that an extra bandwidth smaller than $R^{(u,v)}$ can be found to meet $D_{f^{+}}^{\rm E2E}$. The algorithm seeks to find the minimum extra bandwidth needed to adjust the local deadlines to meet $D_{f^{+}}^{\rm E2E}$ (lines 8-16). Since Eq. (\ref{MiniBand}) indicates that a smaller adjusted local deadline requires a larger bandwidth allocation, the minimum extra bandwidth is achieved by maximizing the adjusted local deadlines so that their sum exactly meets $D_{f^{+}}^{\rm E2E}$. In order to achieve this, the algorithm performs a binary search to determine the minimum consistent ratio $\gamma$ $(0 < \gamma \leq 1)$ such that the sum of the adjusted local deadlines equals $D_{f^{+}}^{\rm E2E}$. In the binary search process, $\gamma$ represents the consistent ratio of extra bandwidth to available residual bandwidth, $\delta$ indicates the adjustment to $\gamma$, and $slack$, initialized to -1, denotes the difference between the end-to-end deadline of the new flow and the sum of the adjusted local deadlines. In each iteration, the \texttt{MapBand2D} function calculates the adjusted local deadlines $\hat{D}_{i}^{(u,v)}$ based on the extra bandwidth $\gamma \cdot R^{(u,v)}$. The $slack$ is then computed, and $\gamma$ is adjusted by $\delta$ based on whether $slack$ is positive or negative, bringing it closer to 0. The binary search process ends when $slack = 0$, at which point the adjusted local deadlines $\hat{D}_{i}^{(u,v)}$ for each $(u,v) \in r_{f^{+}}$ are obtained. The detailed implementations of the \texttt{CalResBand} and \texttt{MapBand2D} functions are as follows.
\begin{algorithm}[t]
	\caption{Local Deadline Adjustment Algorithm}
	\begin{algorithmic}[1]  
		\label{MappingAlgorithm}
		\Statex \hspace*{-14pt} \textbf{Input:} $f^{+}$: flow for addition, $r_{f^{+}}$: candidate route;
		\Statex \hspace*{-14pt} \textbf{Output:} $\{\hat{D}_{i}^{(u,v)}\}_{(u,v)\in r_{f^{+}}}$: adjusted local deadlines for class
		\Statex \hspace*{-11pt}$i$ along route $r_{f^{+}}$; 
		\For{$(u,v) \in r_{f^{+}}$}
		\State $R^{(u,v)} \gets \texttt{CalResBand}()$
		\EndFor
		\For{$(u,v) \in r_{f^{+}}$}
		\State $\hat{D}_{i}^{(u,v),\min} \gets \texttt{MapBand2D}(R^{(u,v)})$
		\EndFor                
		\If{$\sum_{(u,v)\in r_{f^{+}}} \hat{D}_{i}^{(u,v),\min} \leq D_{f^{+}}^{\rm E2E}$}
		\State $\gamma \gets 1$, $\delta \gets 1$, $slack \gets -1$
		\While{$slack \neq 0$}
		\For{$(u,v) \in r_{f^{+}}$}
		\State $\hat{D}_{i}^{(u,v)} \gets \texttt{MapBand2D}(\gamma \cdot R^{(u,v)})$
		\EndFor
		\State $slack \gets D_{f^{+}}^{\rm E2E} - \sum_{(u,v)\in r_{f^{+}}} \hat{D}_{i}^{(u,v)}$
		\State $\delta \gets \delta / 2$
		\State $\gamma \gets (slack > 0) ? \gamma - \delta : \gamma + \delta$
		\EndWhile
		\Else
		\State $\{\hat{D}_{i}^{(u,v)}\}_{(u,v)\in r_{f^{+}}} \gets \text{infeasible}$
		\EndIf
	\end{algorithmic}
\end{algorithm}

\subsubsection{Port-Level Calculating Residual Bandwidth}
\label{CalculateRes}
In line 2 of Algorithm 1, the \texttt{CalResBand} function computes the available residual bandwidth $R^{(u,v)}$ at each egress port $(u,v)$ along the route $r_{f^{+}}$, given by
\begin{equation}
\label{ResidualBand}
R^{(u,v)} = idSl^{\max} - \sum_{j=1}^{N_{\rm AVB}}\bar{idSl}_{j}^{(u,v)},
\end{equation}
where $\bar{idSl}_{j}^{(u,v)}$ denotes the already allocated bandwidth to satisfy the existing local deadlines. Compared to $idSl_{j}^{(u,v)}$ in Eq. (\ref{MiniBand}), which only considers the existing flow set $\mathcal{F}$, $\bar{idSl}_{j}^{(u,v)} \ (j \in [1,N_{\rm AVB}])$ additionally includes the new flow $f^{+}$, and is defined as
\begin{equation}
\label{ddotidSl}
\bar{idSl}_{j}^{(u,v)} = \frac{\sum_{f\in (\mathcal{F} \cup \{f^{+}\})_{j}^{(u,v)} }b_{f}}{D_{j}^{(u,v)}-\frac{l^{\max}}{C} - \frac{(j-1)l^{\max}}{C-\sum_{k=1}^{j-1}\bar{idSl}_{k}^{(u,v)}}},
\end{equation}
which is obtained by replacing $\mathcal{F}$ with $\mathcal{F} \cup \{f^{+}\}$ in the first term of Eq. (\ref{MiniBand}). As described in Algorithm 1, in order to balance the residual bandwidth across egress ports, this available residual bandwidth $R^{(u,v)}$ is used to determine the extra bandwidth required at each egress port to adjust the local deadline $D_{i}^{(u,v)}$.

\subsubsection{Port-Level Mapping Extra Bandwidth to Local Deadline}
\label{SectionMap}
In lines 5 and 11 of Algorithm 1, the \texttt{MapBand2D} function returns the adjusted local deadline $\hat{D}_{i}^{(u,v)}$ when the extra bandwidth is $\gamma \cdot R^{(u,v)} \ (0 < \gamma \leq 1)$. The implementation of the \texttt{MapBand2D} function is detailed in Theorem \ref{TheoadjustD}.

\begin{theorem}[Mapping Extra Bandwidth to Local Deadline] \label{TheoadjustD} Assume that the local deadline for class $i$ is adjusted from \(D_{i}^{(u,v)}\) to \(\hat{D}_{i}^{(u,v)}\), and \(D_{j}^{(u,v)}\) for any other class $j$ remains unchanged. To meet the adjusted local deadlines, the total extra bandwidth that needs to be allocated to all classes at port $(u, v)$ is \(\gamma \cdot R^{(u,v)}\). Then, there exists a mapping from $\gamma \cdot R^{(u,v)}$ to \(\hat{D}_{i}^{(u,v)}\), given by
\begin{equation} 
\begin{aligned} \label{Df} \hat{D}_{i}^{(u,v)} 
&= \frac{\sum_{f\in (\mathcal{F} \cup \{f^{+}\})_{i}^{(u,v)} }b_{f}}{\bar{idSl}_{i}^{(u,v)} + \Phi_{i}} + \frac{l^{\max}}{C} + \frac{(i-1)l^{\max}}{C-\sum_{j=1}^{i-1}\bar{idSl}_{j}^{(u,v)}}, 
\end{aligned} 
\end{equation} 
where \(\Phi_{i}\) denotes the extra bandwidth allocated to class \(i\) that depends on \(\gamma \cdot R^{(u,v)}\), determined by the recursive relationship 
\begin{equation} \sum_{k=i}^{j-1} \Phi_{k} = g\left(\sum_{k=i}^{j} \Phi_{k}\right), \quad \text{for } \, i < j \leq N_{\rm AVB}, 
\end{equation} 
with \(g(\cdot)\) defined in Lemma~\ref{lemma1}, and the initial term is \begin{equation} 
\sum_{k=i}^{N_{\rm AVB}} \Phi_{k} = \gamma \cdot R^{(u,v)}. 
\end{equation}
\end{theorem}

\begin{proof}
For each egress port $(u,v) \in r_{f^{+}}$, the extra bandwidth $\gamma \cdot R^{(u,v)}$, which is based on the already allocated bandwidth $(\bar{idSl})$ in Eq. (\ref{ddotidSl}), is used to meet the reduced local deadline for class $i$. According to Eq. (\ref{ddotidSl}), reducing the local deadline $D_{i}^{(u,v)}$ for class $i$ not only requires allocating extra bandwidth $\Phi_{i}$ to the target class $i$ but also necessitates allocating extra bandwidth $\Phi_{j}$ to all lower-priority classes $j \ (j \in [i+1, N_{\rm AVB}])$ to ensure that their existing local deadlines $D_{j}^{(u,v)}$ remain unaffected by the extra bandwidth allocated to higher-priority classes. Therefore, the extra bandwidth $\gamma \cdot R^{(u,v)}$ is divided among the classes from $i$ to $N_{\rm AVB}$, i.e., $\sum_{k=i}^{N_{\rm AVB}}\Phi_{k} = \gamma \cdot R^{(u,v)}$. The extra bandwidth $\Phi_{i}$ for class $i$ can be obtained through a recursive relationship, which iteratively derives $\sum_{k=i}^{j-1}\Phi_{k}$ from $\sum_{k=i}^{j}\Phi_{k}$ for $i < j \leq N_{\rm AVB}$, beginning with the known value of $\sum_{k=i}^{N_{\rm AVB}}\Phi_{k}$. To ensure that the existing local deadline $D_{j}^{(u,v)}$ for class $j$ is met, the recursive relationship should satisfy $\sum_{k=i}^{j-1}\Phi_{k} = g(\sum_{k=i}^{j}\Phi_{k})$, as proved in Lemma 1. Therefore, based on the extra bandwidth $\Phi_{i}$ for class $i$ obtained from the recursive relationship, the adjusted local deadline $\hat{D}_{i}^{(u,v)}$ is give in Eq. (\ref{Df}), by replacing $\bar{idSl}_{i}^{(u,v)}$ with $\bar{idSl}_{i}^{(u,v)} + \Phi_{i}$ and solving for $\hat{D}_{i}^{(u,v)}$ in Eq. (\ref{ddotidSl}).
\end{proof}

Next, we formally prove Lemma 1, which is used in the proof of Theorem 1. According to Eq. (\ref{ddotidSl}), when adjusting the local deadline \( D_i^{(u,v)} \) of class \( i \), the bandwidth allocation change for higher-priority classes may impact lower-priority class $j \ (j \in [i+1, N_{\rm AVB}])$, potentially causing its existing local deadline \( D_j^{(u,v)} \) to no longer be met. To address this issue, Lemma 1 provides a recursive relationship that must be satisfied when adjusting the local deadline $D_{i}^{(u,v)}$ in order to ensure that \( D_j^{(u,v)} \) remains exactly satisfied.

\begin{lemma}[Recursive Relationship for Adjusting Local Deadline]
	\label{lemma1}
	Let $\Phi_{k}$ denote the extra bandwidth allocated to class $k$ relative to $\bar{idSl}_{k}^{(u,v)}$. Assume that a total extra bandwidth of $\sum_{k=i}^{j-1}\Phi_k$ has already been allocated to classes $i$ through $j-1$ at egress port $(u,v)$, and that the lower-priority class $j$ consequently requires an extra bandwidth of $\Phi_j$ to exactly meet its existing local deadline $D^{(u,v)}_j$. Then, there exists a mapping from $\sum_{k=i}^{j}\Phi_{k}$ to $\sum_{k=i}^{j-1}\Phi_{k}$, given by
	\begin{equation}
	\label{solution}
	\sum_{k=i}^{j-1}\Phi_{k} = g\left(\sum_{k=i}^{j}\Phi_{k}\right) = \frac{-\xi -\sqrt{\xi^{2}-4\cdot \eta \cdot \zeta}}{2\cdot \eta},
	\end{equation}
	where 
	\begin{equation}
	\label{eta}
	\eta = 1 + \frac{\left(C - \sum_{k=1}^{j-1}\bar{idSl}_{k}^{(u,v)}\right)\sum_{f\in (\mathcal{F} \cup \{f^{+}\})_{j}^{(u,v)} }b_{f}}{(j-1)l^{\max}\cdot \bar{idSl}_{j}^{(u,v)}},
	\end{equation}
	\begin{equation}
	\label{xi}
	\xi = -\eta  \sum_{k=i}^{j}\Phi_{k} - \left(\eta-1\right) \left(C - \sum_{k=i}^{j-1}\bar{idSl}_{k}^{(u,v)}\right) - \bar{idSl}_{j}^{(u,v)},
	\end{equation}
	\begin{equation}
	\label{zeta}
	\zeta =  \left(\eta-1\right)  \left(C - \sum_{k=1}^{j-1}\bar{idSl}_{k}^{(u,v)}\right) \sum_{k=i}^{j}\Phi_{k}.
	\end{equation}
\end{lemma}

\begin{proof}
The known $\sum_{k=i}^{j}\Phi_{k}$ can be split into two parts, $\sum_{k=i}^{j-1}\Phi_{k}$ and $\Phi_{j}$, with $\sum_{k=i}^{j}\Phi_{k} = \sum_{k=i}^{j-1}\Phi_{k} + \Phi_{j}$. We first determine the relationship between $\sum_{k=i}^{j-1}\Phi_{k}$ and $\Phi_{j}$ to meet the existing local deadline $D_{j}^{(u,v)}$ for low-prority class $j$. When the extra bandwidth $\sum_{k=i}^{j-1}\Phi_{k}$ has been allocated to classes $i$ through $j-1$, extra bandwidth $\Phi_{j}$ is required for class $j$ to ensure that the existing local deadline $D_{j}^{(u,v)}$ is satisfied. The new bandwidth allocation for class $j$ is obtained by replacing $\sum_{k=1}^{j-1}\bar{idSl}_{k}^{(u,v)}$ with $\sum_{k=1}^{j-1}\bar{idSl}_{k}^{(u,v)} + \sum_{k=i}^{j-1}\Phi_{k}$ in Eq. (\ref{ddotidSl}). Since $\Phi_{j}$ represents the difference between the new bandwidth allocation and the already allocated bandwidth $\bar{idSl}_{j}$, the relationship between $\sum_{k=i}^{j-1}\Phi_{k}$ and $\Phi_{j}$ required to satisfy the existing local deadline $D_{j}^{(u,v)}$ can be expressed as
\begin{equation}
\label{idSlAdd}
\begin{aligned}
\Phi_{j} = \frac{\sum_{f\in (\mathcal{F} \cup \{f^{+}\})_{j}^{(u,v)} }b_{f}}{D_{j}^{(u,v)}-\frac{l^{\max}}{C} - \frac{(j-1)l^{\max}}{C-\left(\sum_{k=1}^{j-1}\bar{idSl}_{k}^{(u,v)} + \sum_{k=i}^{j-1}\Phi_{k}\right)}} - \bar{idSl}_{j}^{(u,v)}.
\end{aligned}
\end{equation}

We then derive the relationship between $\sum_{k=i}^{j-1}\Phi_{k}$ and $\sum_{k=i}^{j}\Phi_{k}$. We use $\sum_{k=i}^{j}\Phi_{k} = \sum_{k=i}^{j-1}\Phi_{k} + \Phi_{j}$ and combine it with the relationship between $\sum_{k=i}^{j-1}\Phi_{k}$ and $\Phi_{j}$ from Eq. (\ref{idSlAdd}) to eliminate $\Phi_{j}$. Thus, we can express the relationship as
\begin{equation}
\sum_{k=i}^{j}\Phi_{k} - \sum_{k=i}^{j-1}\Phi_{k} + \bar{idSl}_{j}^{(u,v)} = \frac{\sum_{f\in (\mathcal{F} \cup \{f^{+}\})_{j}^{(u,v)} }b_{f}}{D_{j}^{(u,v)}-\frac{l^{\max}}{C} - \frac{(j-1)l^{\max}}{C-\left(\sum_{k=1}^{j-1}\bar{idSl}_{k}^{(u,v)} + \sum_{k=i}^{j-1}\Phi_{k}\right)}}.
\end{equation}
Additionally, to reduce the number of symbols used, we eliminate $D_{j}^{(u,v)}$ by solving for it using Eq. (\ref{ddotidSl}) and substituting it accordingly.
This gives the expression as
\begin{equation}
\begin{aligned}
\label{Eq22}
&\sum_{k=i}^{j}\Phi_{k} - \sum_{k=i}^{j-1}\Phi_{k} + \bar{idSl}_{j}^{(u,v)}= \\
&\frac{\sum_{f\in (\mathcal{F} \cup \{f^{+}\})_{j}^{(u,v)} }b_{f}}{\frac{\sum_{f\in (\mathcal{F} \cup \{f^{+}\})_{j}^{(u,v)} }b_{f} }{\bar{idSl}_{j}^{(u,v)}}+\frac{(j-1)l^{\max}}{C-\sum_{k=1}^{j-1}\bar{idSl}_{k}^{(u,v)}} - \frac{(j-1)l^{\max}}{C-\left(\sum_{k=1}^{j-1}\bar{idSl}_{k}^{(u,v)} + \sum_{k=i}^{j-1}\Phi_{k}\right)}}.\\
\end{aligned}
\end{equation}
In Eq. (\ref{Eq22}), $\sum_{k=i}^{j-1}\Phi_{k}$ appears on both sides of the equation. By rearranging the terms, we derive a polynomial equation involving $\sum_{k=i}^{j-1}\Phi_{k}$. Thus, the relationship between $\sum_{k=i}^{j-1}\Phi_{k}$ and $\sum_{k=i}^{j}\Phi_{k}$ to meet the existing local deadline $D_{j}^{(u,v)}$ for low-priority class $j$ is
\begin{equation}
\begin{aligned}
\label{EqForPhi}
\eta \left(\sum_{k=i}^{j-1}\Phi_{k}\right)^{2} + \xi \left(\sum_{k=i}^{j-1}\Phi_{k}\right) + \zeta = 0,
\end{aligned}
\end{equation} 
where the coefficients $\eta$, $\xi$, and $\zeta$ associated with  $\sum_{k=i}^{j}\Phi_{k}$ are given in Eqs. (\ref{eta}), (\ref{xi}) and (\ref{zeta}), respectively. Then, we attempt to find the solution for $\sum_{k=i}^{j-1}\Phi_{k}$ in Eq. (\ref{EqForPhi}) within the feasible interval $[0, \sum_{k=i}^{j}\Phi_{k}]$, since $\sum_{k=i}^{j-1}\Phi_{k}$ cannot exceed $\sum_{k=i}^{j}\Phi_{k}$. We let $f(x) = \eta x^{2} + \xi x + \zeta$, and observe that $\eta > 0$, $f(0) = \zeta \geq 0$, and $f(\sum_{k=i}^{j}\Phi_{k}) = - \bar{idSl}_{j}^{(u,v)} \cdot \sum_{k=i}^{j}\Phi_{k}\leq 0$. Therefore, of the two solutions to Eq. (\ref{EqForPhi}), one lies within the interval $[0, \sum_{k=i}^{j}\Phi_{k}]$ and the other within $[\sum_{k=i}^{j}\Phi_{k}, \infty)$. Thus, Eq. (\ref{EqForPhi}) yields a unique solution within the feasible interval $[0, \sum_{k=i}^{j}\Phi_{k}]$, representing the unique recursive relationship to obtain $\sum_{k=i}^{j-1}\Phi_{k}$ from the known $\sum_{k=i}^{j}\Phi_{k}$, as given by Eq. (\ref{solution}).
\end{proof}

\section{Experimental Comparison and Analysis}
\label{Experiment}
In this section, we evaluate the performance of our method in synthetic test cases and realistic test cases. The experiments are implemented in C++ and run on an Intel(R) Core(TM) i5-1240P 64bit CPU @1.7GHz with 16GB of RAM.

\subsection{Synthetic Test Cases}

\subsubsection{Experiment Case Setup}
\textcolor{black}{The synthetic test cases are generated using different network topologies and request flow sets.} Network topologies are created using the Erdős-Rényi (ER) model through NetworkX \cite{hagberg2008}, a Python library dedicated to creating complex networks. The ER model independently includes each possible link connection between pairs of switches with a probability $p$. We adopt link connection probabilities $p$ of $\{0.4, 0.5, 0.6, 0.7, 0.8 \}$. Additionally, we use various topology scales of $\{$10SW50ES, 14SW70ES, 18SW90ES, 22SW110ES$\}$. An example of the 22SW110ES topology, which includes 22 switches (SWs) and 110 end systems (ESs), with a link connection probability of $p=0.6$, is shown in Fig. \ref{Topology}. The transmission rate of physical links is set to $C=$100Mbit/s. Requested flow sets are generated in quantities of $\{400,500,600,700,800\}$. Packet sizes for these flows range from 64 to 1518 bytes. To express the adaptability of the method across different periods and deadlines, flow periods and deadlines are set independently, ranging from 2ms to 9ms for each. Source and destination nodes are randomly selected from the end-systems, ensuring that the source and destination of a flow are not connected to the same switch. The number of AVB classes $N_{\rm AVB}$ is chosen from $\{1,2,4,8\}$, with flows distributed evenly across different classes according to their deadlines. The upper limit of bandwidth allocated to all AVB classes at the egress port, denoted as $idSl^{\max}$, is set to $0.75C$. The maximum frame size for BE traffic is set to 1518 bytes. The parameter $k$ for the k-shortest candidate paths is set to 3. In order to evaluate the performance of the admission control method, we set the offline flows to zero so that all flows are dynamically admitted by the method. For each class level, the initial local deadline is uniformly set as the ratio of the maximum end-to-end deadline to the minimum route length within the class. For example, when there are two AVB classes, the initial local deadlines are 5/2 = 2.5 ms and 9/2 = 4.5 ms, respectively, where 5 ms and 9 ms are the maximum end-to-end deadlines for class 1 and class 2, and 2 is the minimum route length. The initial bandwidth allocations are uniformly set to zero.
\begin{figure}[!t]
	\centerline{\includegraphics[width=\columnwidth]{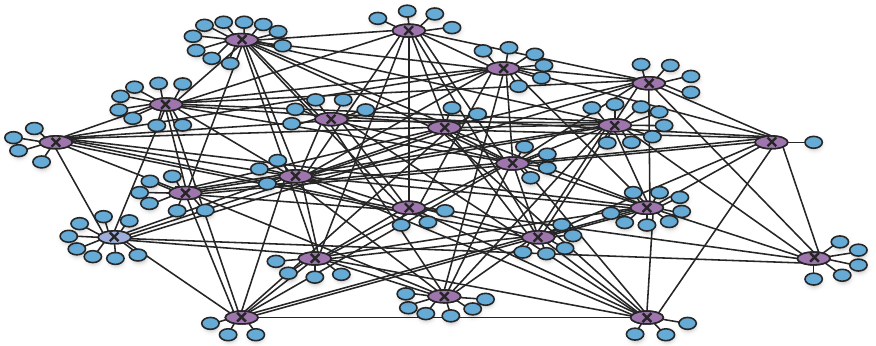}}
	\caption{An example of an Erdős-Rényi (ER) network with 22 switches, 110 end-systems, and a link connection probability of 0.6.}
	\label{Topology}
\end{figure}
\begin{figure}[!t]
	\centerline{\includegraphics[width=\columnwidth]{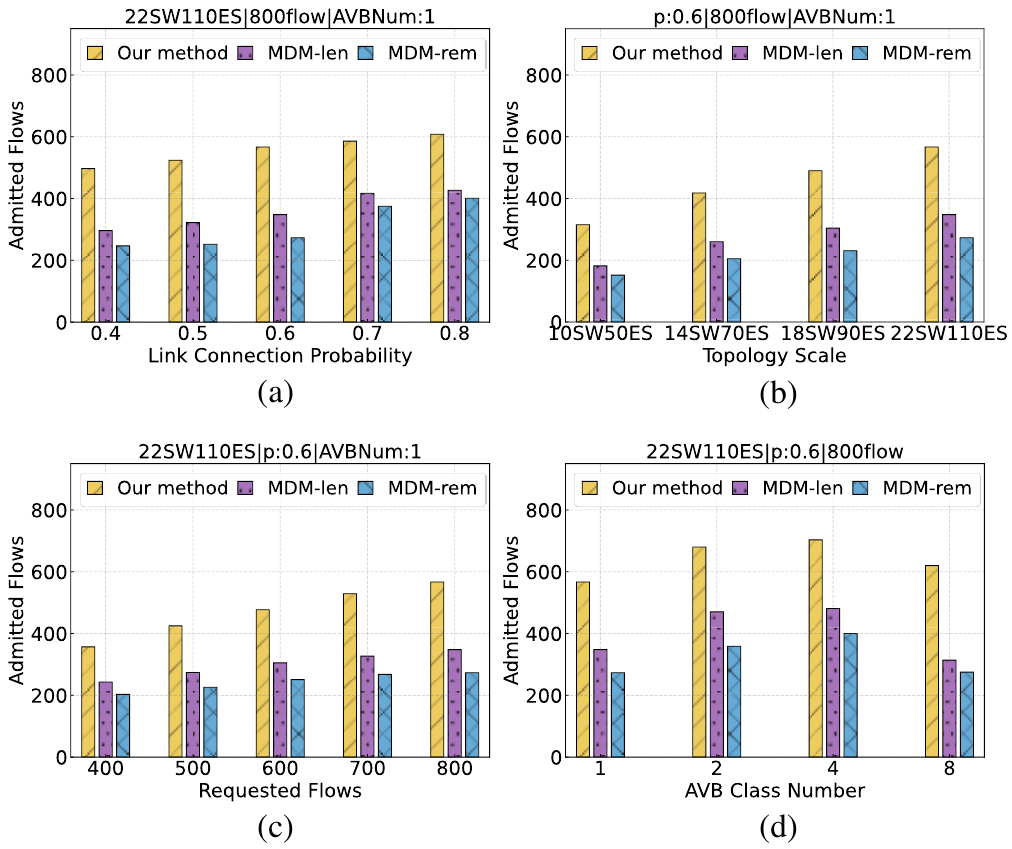}}
	\caption{Comparison of our method with the state-of-the-art in terms of admitted flows.}
	\label{Lisa}
\end{figure}
\subsubsection{Comparison with the State-of-the-Art}
\label{Exp1Section}
To verify the high utilization, we compare our method with the state-of-the-art, namely MDM-len and MDM-rem \cite{maile2022}, which are the online admission control methods in TSN/CBS architecture. These methods ensure real-time requirements by setting fixed delay budgets on queues and employing DCLC routing at runtime, with MDM-len using a cost function based on path length and MDM-rem using a cost function based on the remaining link rate. The evaluations compare the number of admitted flows across experiment cases varying in link connection probabilities, topology scales, requested flows, and AVB class number. Fig. \ref{Lisa}(a)-(d) show the comparative results. Fig. \ref{Lisa}(a) shows the results for different link connection probabilities $p$. Our method exhibits greater improvement over the MDM method when the link connection probability $p$ is smaller. This is because, at a smaller link connection probability $p$, the availability of network resources is lower. Our method can better overcome these limitations compared to fixed delay budgets by adaptively assigning and adjusting local deadlines,  thus effectively increasing the flow admission capacity. Fig. \ref{Lisa}(b) indicates that as the topology scale expands, the number of admitted flows increases for all methods. This growth is due to expanding available network resources with increased physical links. Fig. \ref{Lisa}(c) shows that as the number of requested flows increases, the admission ratio, which represents the percentage of admitted flows to requested flows, decreases for all methods due to the limited available network resources. Fig. \ref{Lisa}(d) illustrates that as the number of AVB classes increases, the number of admitted flows first increases and then decreases. This trend suggests that assigning too many or too few classes might not be optimal, making class assignments an interesting topic for future discussion. In summary, across these experiment cases, our method shows an average improvement of 59.4\% over the MDM-len method and 95.2\% over the MDM-rem method in terms of admitted flows.

\subsubsection{Comparison of Local Deadline Adjustment Strategies}
\label{ExpStra}
To validate the effectiveness of our local deadline adjustment strategy proposed in Section \ref{Strategy}, we compare it with three baseline strategies from \cite{zotero-5194}: equal partition (EP), load-based partition (LP), and available bandwidth-based Partition (ABP). The coefficients $\kappa^{(u,v)}$ for $(u,v) \in r_{f^{+}}$ in these strategies are listed in Table \ref{CompStraTable}. For the comparison, we adopt the adjusted local deadline $\hat{D}^{(u,v)}_{i} = D^{(u,v)}_{i} - (\sum_{(u',v') \in r_{f^{+}}}D^{(u',v')}_{i} - D_{f^{+}}^{\rm E2E}) \cdot \kappa^{(u,v)}$ for the baseline strategies in place of the output from Algorithm 1, while the other parts follow the online admission control framework outlined in Section \ref{FlowAdd}. The evaluations compare the number of admitted flows across experiment cases, which are consistent with those used in Section \ref{Exp1Section}. Fig. \ref{EAPA}(a)-(d) displays the experimental results, indicating that our strategy achieves an average increase in the number of admitted flows by 40.0\%, 32.0\%, and 34.7\% compared to the EP, LP, and ABP strategies, respectively. The EP strategy leads to the lowest number of admitted flows because it uniformly reduces local deadlines to meet the new flow deadline without considering the specific characteristics of the existing network configuration. The LP and ABP outperform EP because they indirectly incorporate residual bandwidth balancing into their strategies by utilizing the proportions of traffic load and available residual bandwidth, respectively. Our strategy outperforms others because it iteratively finds an adjustment that ensures the extra bandwidths are allocated according to the proportion of available residual bandwidths between egress ports. This approach directly balances the residual bandwidth across egress ports, preventing some egress ports from becoming bottlenecks too early due to resource exhaustion.
\begin{table}[!t]
	\begin{footnotesize}
		\begin{threeparttable}
			\centering
			\caption{Coefficients of Baseline Strategies}
			\label{CompStraTable}
			\setlength{\tabcolsep}{3pt}
			\renewcommand{\arraystretch}{0.2}
			\begin{tabular}{>{\centering\arraybackslash}m{30pt} >{\centering\arraybackslash}m{19pt} >{\centering\arraybackslash}m{95pt} >{\centering\arraybackslash}m{65pt}}
				\hline
				\\
				Strategy&EP&LP&ABP\\
				\hline		
				\\
				${{\kappa^{(u,v)}}}^{1}$&${\frac{1}{|r_{f^{+}}|}}^{2}$&${\frac{\sum_{(u',v') \in r_{f^{+}}} B^{(u',v')} -B^{(u,v)} }{\left(|r_{f^{+}}| - 1\right)\sum_{(u',v')\in r_{f^{+}}}B^{(u',v')}}}^{3}$&${\frac{R^{(u,v)}}{\sum_{(u',v') \in r_{f^{+}}}R^{(u',v')}}}^{4}$\\
				\\
				\hline
			\end{tabular}
			\begin{tablenotes}		
				\item[1] 
				$\kappa^{(u,v)}$ denotes the proportion of the reduction in the existing local deadline at egress port $(u,v)$ relative to the entire route $r_{f^{+}}$, where $\hat{D}^{(u,v)}_{i} = D^{(u,v)}_{i} - (\sum_{(u',v') \in r_{f^{+}}}D^{(u',v')}_{i} - D_{f^{+}}^{\rm E2E}) \cdot \kappa^{(u,v)}$.
				
				\item[2] $|r_{f^{+}}|$ denotes the length of route $r_{f^{+}}$. 
				
				\item[3] $B^{(u,v)} = \sum_{f \in  (\mathcal{F} \cup \{f^{+}\})^{(u,v)}}\rho_{f}$ denotes the load at egress port $(u,v)$. 
				
				\item[4] $R^{(u,v)}$ denotes the residual bandwidth at egress port $(u,v)$ as Eq. (\ref{ResidualBand}).
			\end{tablenotes}
		\end{threeparttable}
	\end{footnotesize}
\end{table}
\begin{figure}[!t]
	\centerline{\includegraphics[width=\columnwidth]{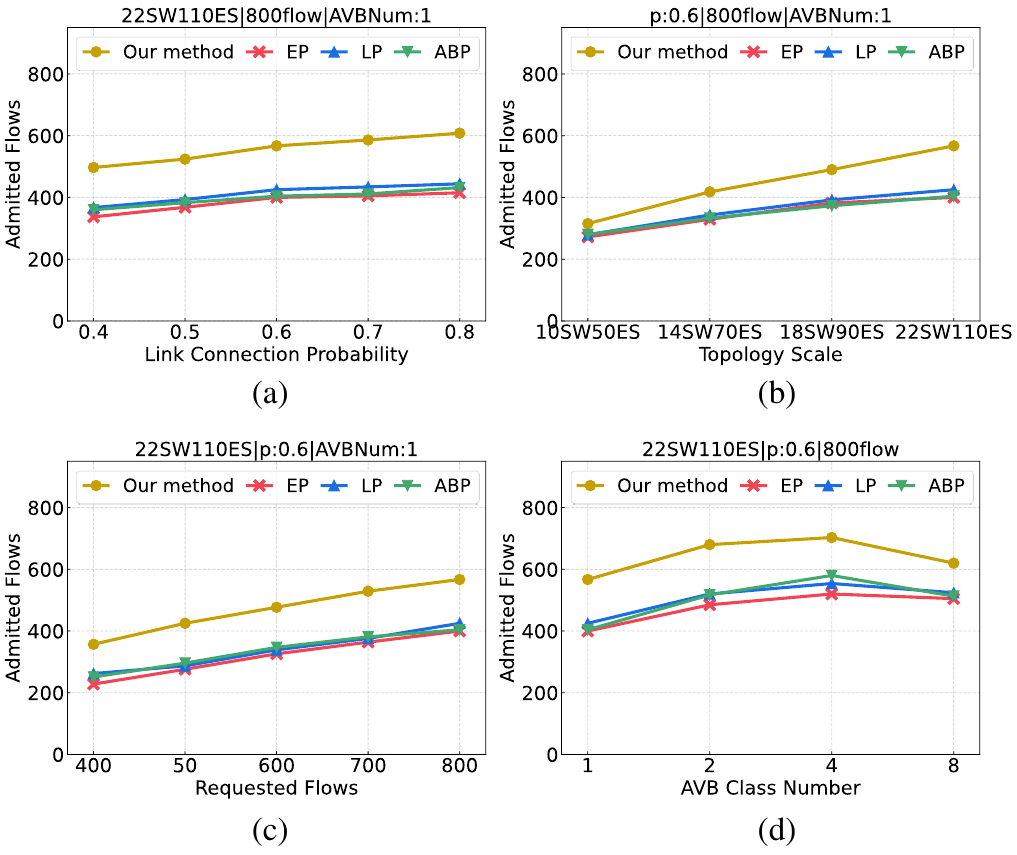}}
	\caption{Comparison of our local deadline adjustment strategy in Section \ref{DelayAdjust} with the baseline strategies in terms of admitted flows.}
	\label{EAPA}
\end{figure}

\subsubsection{Execution Time}
\begin{figure}[!t]
	\centerline{\includegraphics[width=\columnwidth]{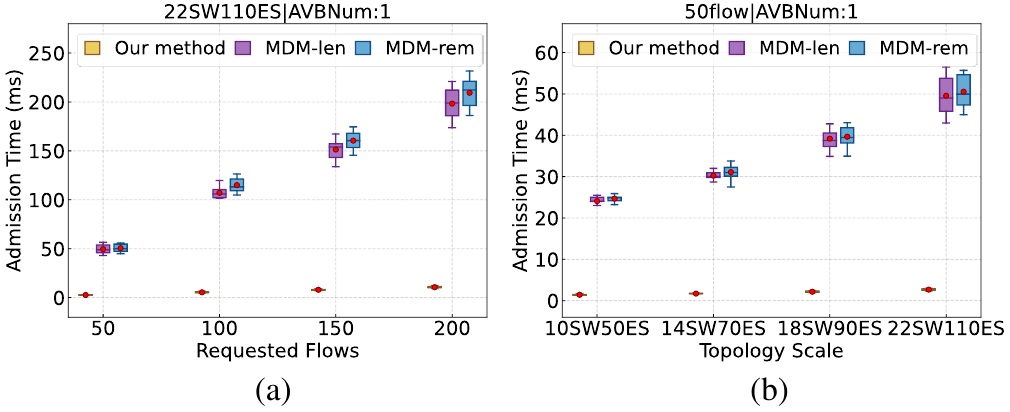}}
	\caption{Comparison of our method with the state-of-the-art in terms of admission time.}
	\label{RunTime}
\end{figure}
To confirm the efficiency and scalability, we compare our method with the state-of-the-art MDM-len and MDM-rem\cite{maile2022}. The evaluations focus on the admission time across experimental cases with varying requested flows and topology sizes. Fig. \ref{RunTime}(a)-(b) demonstrate the experimental results using box plots, where each box corresponds to 10 cases and the red marker represents the mean value. Fig. \ref{RunTime}(a) illustrates that the admission time for each method increases linearly with the number of requested flows. Specifically, for our method, the average admission times for admitting 50, 100, 150, and 200 flows are 2.7ms, 5.4ms, 8.0ms, and 10.6ms, respectively. Compared to the MDM methods, our method exhibits a markedly slower growth rate, indicating a significant decrease in the per-flow admission time. Fig. \ref{RunTime}(b) displays the admission time for three methods across various topology scales. For our method, the average admission times for admitting 50 flows are 1.4ms, 1.7ms, 2.2ms, and 2.7ms, respectively. Compared to the MDM methods, our method demonstrates a smaller rise in admission time as the topology scale grows. In summary, for the presented cases, our method achieves an average reduction in admission time of 94.5\% and 94.7\% when compared to the MDM-len and MDM-rem methods, respectively. The per-flow admission time for large-scale scenarios involving 22SW110ES achieves less than 100$\mu$s. Our method significantly reduces admission time for two main reasons: Firstly, we shift the complexity of online routing to the offline phase. Secondly, although our method involves online dynamic local deadline adjustments, it utilizes analytical results to reduce the solution space and accelerate the solving process without introducing significant computational overhead. These improvements collectively lead to a substantial reduction in admission time.

\subsection{Realistic Test Cases}
\begin{figure}[!t]
	\centerline{\includegraphics[width=\columnwidth]{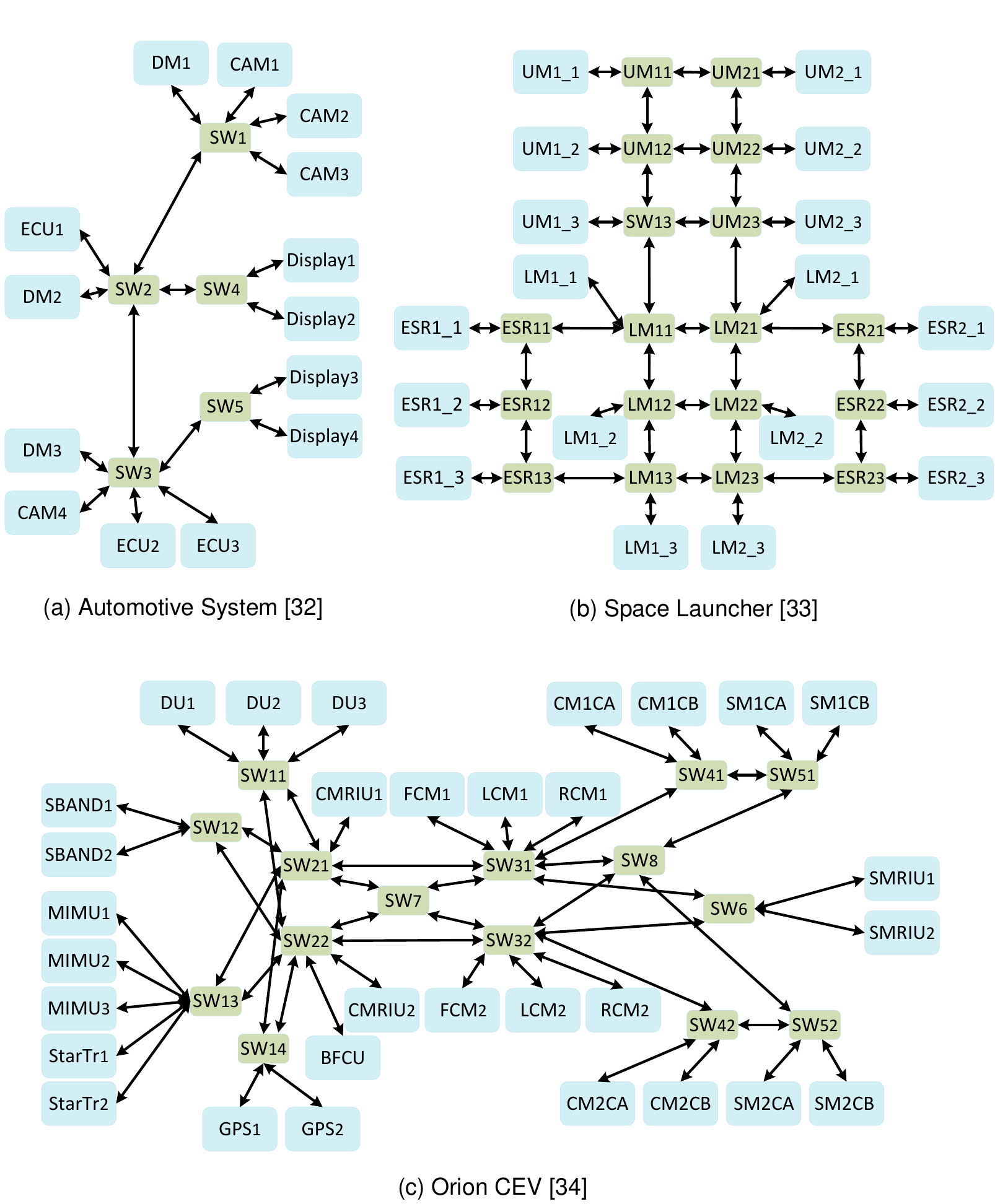}}
	\caption{Topologies of three realistic test cases.}
	\label{TopoReal}
\end{figure}
\subsubsection{Comparison in Realistic Test Cases}
We compare different methods across three realistic test cases. The methods consist of our method, MDM-len and MDM-rem (described in Section \ref{Exp1Section}), and EP, LP, and ABP (described in Section \ref{ExpStra}). The three realistic test cases are the Automotive System \cite{migge2018}, Space Launcher \cite{keller2022}, and Orion CEV \cite{tamas-selicean2014a}. Their topologies and network configuration characteristics are presented in Fig. \ref{TopoReal} and Table \ref{OverviewTable}. The attributes of the requested flow sets used for admission control are listed in Tables \ref{AutoTable}, \ref{SpaceTable}, and \ref{OrionTable}, which are derived from real-world cases. The comparison metrics include admission capacity, first rejection, and per-flow admission time. Admission capacity refers to the number of flows successfully admitted when processing a large set of request flows, reflecting the maximum bandwidth utilization of the method. The first rejection indicates the index of the first flow request rejected during the admission process, with a higher value suggesting that the method can better adapt to varying traffic requirements. Per-flow admission time refers to the average computation time required to admit each flow, reflecting the responsiveness to flow requests. The comparison results are summarized in Table \ref{tab:all_cases}.
\begin{table}[!t]
	\centering
	\renewcommand{\arraystretch}{1} 
	\caption{Overview of Network Configuration Characteristics}
	\begin{footnotesize}
		\begin{tabular}{|>{\centering\arraybackslash}m{1.65cm}|c|c|c|>{\centering\arraybackslash}m{2.5cm}|}
			\hline
			Case              & Nodes & Switches & Links & Link Speeds \\ \hline
			Automotive System \cite{migge2018} & 14    & 5        & 18    & \makecell{1Gbit/s (DM3$\leftrightarrow$ SW3) \\ 100Mbit/s(others)}\\ \hline
			Space Launcher \cite{keller2022}    & 18    & 18       & 24   & 100Mbit/s \\ \hline
			Orion CEV    \cite{tamas-selicean2014a}     & 31    & 15       & 55   & 1Gbit/s\\ \hline
		\end{tabular}
	\end{footnotesize}
	\label{OverviewTable}
\end{table}	

First, our method is compared with MDM-len and MDM-rem methods. In terms of admission capacity, our method shows improvements across all three test cases, with a remarkable increase of 70\% observed in the Space Launcher case. For the first rejection, our method also demonstrates improvement, with performance increasing by up to 75\% in the Orion CEV case. Regarding per-flow admission time, our method reduces the time by over 79\% compared to MDM-len. Additionally, even in large-scale scenarios, the per-flow admission time of our method is less than 100 $\mu s$. These results indicate that, with ATS support, our method dynamically adjusts local deadlines to balance residual bandwidth, thereby optimizing resource utilization and enhancing adaptability. Furthermore, our method avoids online routing and reduces computational overhead during the adjustment process by leveraging analytical results to reduce the solution space, thereby improving admission time.
\begin{table}[!t]
	\centering
	\begin{footnotesize}
		\caption{Flows in the Automotive System Case}
		\begin{tabular}{|>{\centering\arraybackslash}m{0.5cm}|>{\centering\arraybackslash}m{1.1cm}|>{\centering\arraybackslash}m{0.8cm}|>{\centering\arraybackslash}m{0.85cm}|>{\centering\arraybackslash}m{0.65cm}|>{\centering\arraybackslash}m{0.9cm}|>{\centering\arraybackslash}m{1.05cm}|}
			\hline
			Class              & \makecell{Size\\(B)}  & \makecell{Period\\(us)} & \makecell{Deadline\\(us)} & Src & Dst & Proportion \\ \hline
			A                  & 256-1024 & 10000      & 10000        & ECU    & DM          & 40.8\%     \\ \hline
			B                  & 128or256 & 10000      & 10000        & CAM    & DM          & 29.6\%     \\ \hline
			\multirow{6}{*}{C} & 1446     & 1100       & 30000        & CAM1   & DM1         & 3.7\%      \\ \cline{2-7} 
			& 1446     & 1100       & 30000        & CAM2   & DM1         & 3.7\%      \\ \cline{2-7} 
			& 1446     & 1100       & 30000        & CAM3   & DM2         & 3.7\%      \\ \cline{2-7} 
			& 1446     & 1100       & 30000        & CAM4   & DM2         & 3.7\%      \\ \cline{2-7} 
			& 1446     & 1100       & 30000        & CAM4   & Display1    & 3.7\%      \\ \cline{2-7} 
			& 1446     & 1100       & 30000        & CAM4   & Display2    & 3.7\%      \\ \hline
			D                  & 1446     & 1100       & 30000        & CAM4   & DM3         & 7.4\%      \\ \hline
		\end{tabular}
		\label{AutoTable}
	\end{footnotesize}
\end{table}

\begin{table}[!t]
	\centering
	\begin{footnotesize}
		\caption{Flows in the Space Launcher Case}
		\begin{tabular}{|c|c|c|c|c|}
			\hline
			Class & Size(B)  & Period(us) & Deadline(us) & Proportion \\ \hline
			A     & 256-1024 & 10000      & 10000        & 21\%       \\ \hline
			B     & 128or256 & 10000      & 10000        & 78\%       \\ \hline
			C     & 1446     & 1100       & 30000        & 1\%        \\ \hline
		\end{tabular}
		\label{SpaceTable}
	\end{footnotesize}
\end{table}

\begin{table}[!t]
	\centering
	\begin{footnotesize}
		\caption{Flows in the Orion CEV Case}
		\begin{tabular}{|c|c|c|c|c|}
			\hline
			Class & Size(B) & Period (us) & Deadline (us) & \begin{tabular}[c]{@{}c@{}} Proportion\end{tabular} \\ \hline
			\multirow{3}{*}{A} & 64-1518& 4000   & 6800-7300   & 6\%  \\ \cline{2-5} 
			&64-1518& 8000   & 8700-15000   & 14\%  \\ \cline{2-5} 
			&64-1518& 16000   & 16000-30000   & 9\%  \\ \hline
			\multirow{3}{*}{B} &64-1518& 16000  & 17000-32000  & 13\% \\ \cline{2-5} 
			&64-1518& 32000   & 34000-62000   & 14\%  \\ \cline{2-5} 
			&64-1518& 64000   & 67000-68000   & 2\%  \\ \hline
			\multirow{2}{*}{C} &64-1518& 64000  & 70000-130000  & 24\% \\ \cline{2-5} 
			&64-1518& 128000 & 170000-190000 & 3\% \\ \hline
			D                  &64-1518& 128000      & 170000-370000      & 15\%    \\ \hline
		\end{tabular}
		\label{OrionTable}
	\end{footnotesize}
\end{table}	

\begin{table*}[!ht]
	\centering
	\renewcommand{\arraystretch}{2.26} 
	\setlength{\tabcolsep}{8pt} 
	\begin{footnotesize}
		\begin{threeparttable}
			\caption{Comparison in Three Realistic Test Cases}
			\begin{tabular}{>{\centering\arraybackslash}m{2.4cm}|>{\centering\arraybackslash}m{3.2cm}|>{\centering\arraybackslash}m{1.55cm}|>{\centering\arraybackslash}m{1.4cm}|>{\centering\arraybackslash}m{1.4cm}|>{\centering\arraybackslash}m{1.2cm}|>{\centering\arraybackslash}m{1.2cm}|>{\centering\arraybackslash}m{1.2cm}}
				\Xhline{1pt}
				\textbf{Case} & \textbf{Metric}           & \textbf{Our Method} & \textbf{MDM-len} & \textbf{MDM-rem} & \textbf{EP} & \textbf{LP} & \textbf{ABP} \\ \Xhline{1pt}
				\multirow{3}{*}{\makecell{\textbf{Automotive System} \\ \cite{migge2018}}} 
				& Admission Capacity$^{1}$      & \makecell{217 \\ (+28\%)$^{4}$}       & \makecell{170 \\  (0\%)}      & \makecell{170 \\ (0\%)}      & \makecell{170 \\ (0\%)}  & \makecell{194 \\ (+14\%)}  & \makecell{126 \\ (-26\%)}  \\ \cline{2-8} 
				& First Rejection$^{2}$            & \makecell{48 \\ (+2\%)}       &\makecell{ 47  \\ (0\%) }    &\makecell{47 \\ (0\%)}      & \makecell{48\\ (+2\%)} & \makecell{48 \\(+2\%)} & \makecell{6 \\(-87\%)} \\ \cline{2-8} 
				& Per-Flow Admission Time$^{3}$ & \makecell{13.3$\mu s$ \\ (-90\%)}       & \makecell{137.7$\mu s$ \\ (0\%)}        &\makecell{ 90.4$\mu s$ \\ (-34\%) }      & \makecell{11.1$\mu s$ \\ (-92\%)}   & \makecell{9.67$\mu s$ \\ (-93\%)} & \makecell{11.1$\mu s$ \\ (-92\%)} \\ \hline \hline
				\multirow{3}{*}{\makecell{\textbf{Space Launcher}\\ \cite{keller2022}}} 
				& Admission Capacity    & \makecell{701  \\ (+70\%)}       & \makecell{412  \\ (0\%)  }     & \makecell{418\\ (+1\%) }      & \makecell{557 \\ (+35\%)}  & \makecell{540 \\ (+31\%) } & \makecell{383 \\ (-7\%)}  \\ \cline{2-8} 
				& First Rejection          & \makecell{218  \\ (+37\%)   }     & \makecell{159    \\ (0\%) }    &\makecell{ 121 \\ (-24\%)  }      & \makecell{9\\ (-94\%) }& \makecell{9 \\ (-94\%)}&\makecell{ 9\\ (-94\%) }\\ \cline{2-8} 
				& Per-Flow Admission Time & \makecell{70.0 $\mu s$ \\ (-88\%)}       & \makecell{565.4$\mu s$ \\ (0\%)}    & \makecell{581.6$\mu s$ \\ (+2\%)}        &  \makecell{30.1$\mu s$ \\ (-95\%)}& \makecell{56.7$\mu s$ \\ (-90\%)}  & \makecell{69.5$\mu s$ \\ (-88\%)} \\ \hline \hline
				\multirow{3}{*}{\makecell{\textbf{Orion CEV} \\ \cite{tamas-selicean2014a}}} 
				& Admission Capacity      & \makecell{5653 \\ (+18\%)    }    &\makecell{ 4789  \\ (0\%)  }    &\makecell{ 4293  \\ (-10\%) }     & \makecell{2000\\ (-58\%) }& \makecell{2126 \\ (-56\%)} & \makecell{1969 \\ (-59\%)}\\ \cline{2-8} 
				& First Rejection           &\makecell{ 3588 \\ (+75\%)}       &\makecell{ 2047    \\ (0\%)}   & \makecell{1449 \\ (-29\%)}      & \makecell{21\\ (-99\%) }& \makecell{21 \\ (-99\%)}& \makecell{21\\ (-99\%)} \\ \cline{2-8} 
				& Per-Flow Admission Time & \makecell{78.3$\mu s$ \\ (-79\%)}  & \makecell{370.1$\mu s$ \\ (0\%) }       & \makecell{420.2$\mu s$ \\ (+14\%)}       & \makecell{24.5$\mu s$   \\ (-93\%)}      &\makecell{ 25.7$\mu s$ \\ (-93\%)}& \makecell{32.2$\mu s$ \\ (-91\%)}  \\ \Xhline{1pt}
			\end{tabular}
			\label{tab:all_cases}
			\begin{tablenotes}		
				\item[1] Admission capacity is evaluated using a request flow set of 10000 flows, representing a large number of admission requests.
				
				\item[2] First rejection refers to the index of the first flow request rejected during the admission process. 
				
				\item[3] Per-flow admission time is measured using a request flow set of 100 flows, representing a high admission rate, which helps avoid underestimated results caused by frequent rejections.
				
				\item[4] The baseline for comparison is the state-of-the-art MDM-len method.
			\end{tablenotes}
		\end{threeparttable}
	\end{footnotesize}
\end{table*}

Next, our method is compared with EP, LP, and ABP, which use different local deadline adjustment strategies. For admission capacity, our method demonstrates improvements across all three test cases, with an increase of over 165\% in the Orion CEV case. In terms of first rejection, our method performs similarly to the EP and LP methods in the Automotive System case but shows significant improvements in the large-scale scenarios of Space Launcher and Orion CEV. In the Space Launcher case, the first rejection for EP, LP, and ABP occurs at the 9th flow, while our method postpones it until the 218th flow. In the Orion CEV case, the first rejection for EP, LP, and ABP happens at the 21st flow, while our method postpones it until the 3588th flow. Regarding per-flow admission time, our method shows an increase compared to EP, LP, and ABP, as the bandwidth-balancing local deadline adjustment strategy employed in our method introduces additional computational overhead. Despite the increase, the admission times remain within the same order of magnitude (e.g., tens of microseconds), as our method reduces the solution space through analytical results, which accelerates the solving process. In summary, compared to the EP, LP, and ABP methods that directly adjust local deadlines, our method adjusts local deadlines by balancing the remaining bandwidth. Although this introduces some additional admission time overhead, it significantly improves admission capacity and postpones the occurrence of the first rejection of the flow requests. Particularly in large-scale scenarios, our method significantly increases the number of consecutive admissions, demonstrating enhanced adaptability.

\subsubsection{In-Depth Comparison in the Orion CEV Case}

We conduct an in-depth comparison of different methods in the Orion CEV case, focusing on the distribution of successful admitted flows and bottleneck egress port counts during the admission process. The experiment uses a request flow set comprising 6000 flows, which are divided into 120 flow groups, each containing 50 flows. For each flow group, two metrics are recorded when the admission decision for the last flow in the group is completed: the number of successful admitted flows within the group and the number of bottleneck egress ports in the network. A bottleneck egress port is defined as an egress port where the available residual bandwidth is less than 10\% of $idSl^{\max}$, where $idSl^{\max}$ is the upper bandwidth limit allocated to all AVB classes at the egress port. The comparison results are shown in Fig. \ref{index}, where each subplot corresponds to a different method. In each subplot, the x-axis represents the index of the flow groups requesting admission. The blue points indicate the number of successful admitted flows in the corresponding flow group (ranging from 0 to 50), and the red points represent the number of bottleneck egress ports.

From the comparison results in Fig. \ref{index}, it is observed that our method postpones the appearance of bottleneck egress ports until the 65th flow group. In comparison, MDM-len reaches this point in the 26th group, MDM-rem in the 24th group, EP in the 13th group, LP in the 9th group, and ABP in the 8th group. Further observation of the admitted flows reveals that, after the appearance of bottleneck egress ports, our method, MDM-len, and MDM-rem begin to fail to admit all 50 flows as the egress ports gradually become unavailable. Additionally, as the number of bottleneck egress ports increases, the admission ratio within each flow group shows a decreasing trend. Overall, our method enhances admission capacity and delays the first rejection of flow requests by effectively postponing the occurrence of bottleneck egress ports. The EP, LP, and ABP methods are unable to admit all 50 flows even before bottleneck egress ports appear. This is because the local deadlines determined by their direct deadline adjustment strategies cannot be guaranteed within the available residual bandwidth, leading to the rejection of some flows. In contrast, our method avoids this issue by using a local deadline adjustment strategy that balances the residual bandwidth, thereby postponing the first rejection of flow requests and achieving higher admission capacity.
\begin{figure*}[!t]
	\centerline{\includegraphics[width=\textwidth]{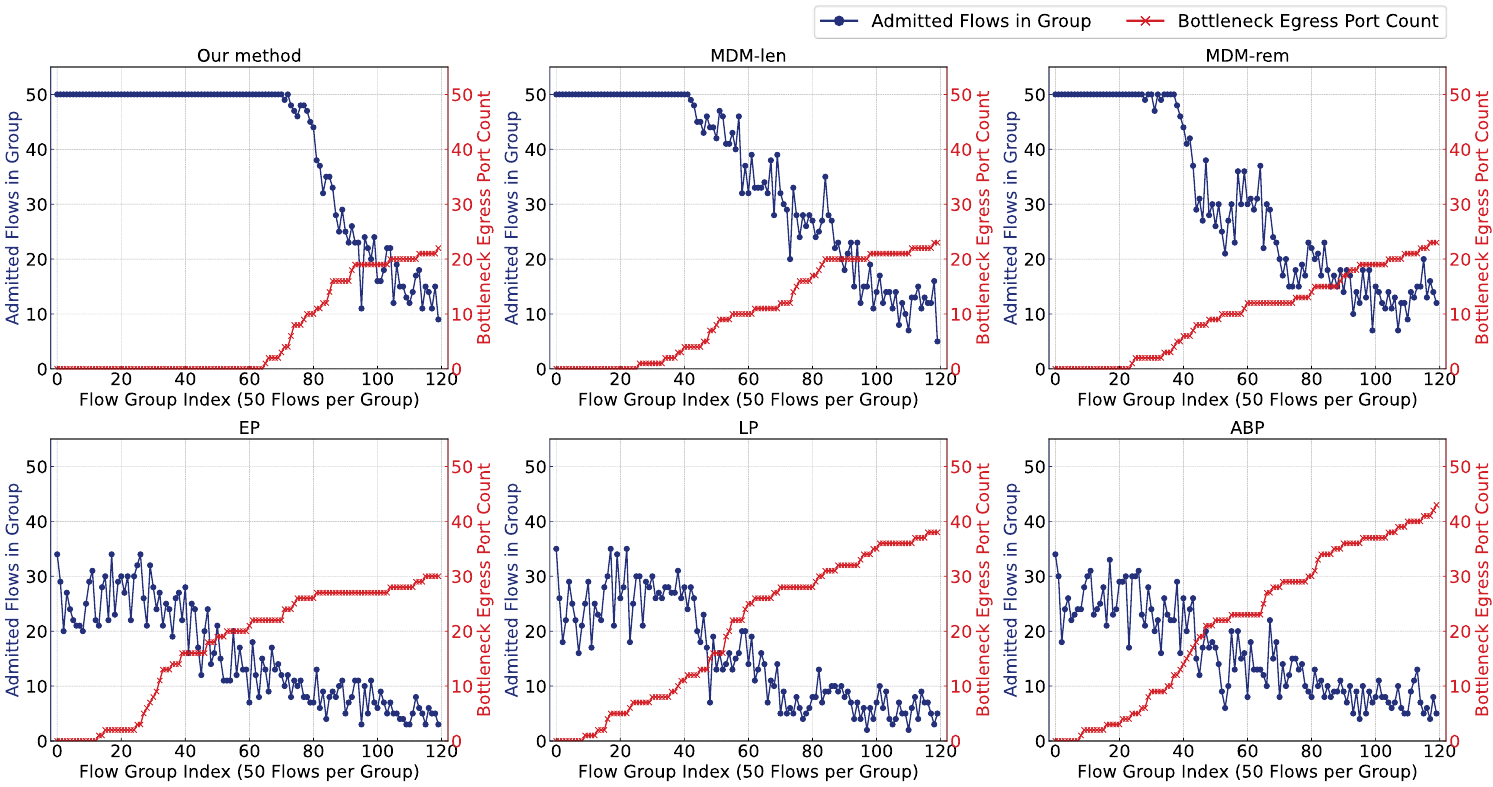}}
	\caption{Comparison of distribution in successful admission flows and bottleneck egress port count in Orion CEV case.}
	\label{index}
\end{figure*}

\section{Conclusion}
\label{Conclusion}
In this work, we propose a rapid, scalable, and high-utilization online admission control method for time-critical ET traffic based on the TSN/ATS+CBS architecture. Experimental results from both synthetic and realistic test cases demonstrate that, compared to the state-of-the-art, our method increases the number of admitted flows by an average of 56\% and reduces admission time by an average of 92\%. Additionally, evaluations show that our method postpones the occurrence of bottleneck egress ports and the first rejection of admission requests, thus improving adaptability.

\bibliographystyle{elsarticle-num}      
\bibliography{online}

\vfill
\end{document}